\documentclass[10pt, onecolumn]{article}

\usepackage{authblk}

\usepackage{amsmath}

\usepackage{amssymb}
\usepackage{amsthm}
\usepackage[margin = 1in]{geometry}
\usepackage{cite}
\usepackage{booktabs}
\usepackage{graphicx}
\usepackage{setspace}
\usepackage{url}
\usepackage{epstopdf}
\usepackage{verbatim}
\usepackage{empheq}
\usepackage{slashbox}
\usepackage{gensymb}

\renewcommand{\le}{ \leqslant }
\renewcommand{\ge}{ \geqslant }

\newcommand{\lam}{\lambda}

\newcommand{\eps}{\epsilon}

\newcommand{\half}{\displaystyle \frac{1}{2}}

\newcommand{\real} {\mathbb{R}}

\newcommand{\Kset} {\mathcal{K}}		%

\newcommand{\opt}{^{\mathsf{*}}}

\newcommand{\bnorm }{\theta}			%

\newcommand{\ds} {\displaystyle}

\newcommand{\figurescale}		{0.750}				

\allowdisplaybreaks

\providecommand{\norm}[1]{\lVert#1\rVert}
\providecommand{\abs}[1]{\left\vert #1 \right\vert}

\providecommand{\tabref}[1]{Table~\ref{#1}}

\providecommand{\figref}[1]{Fig.~\ref{#1}}
\providecommand{\eqref}[1]{\eqref{#1}}

\title{Sparsity-based~Algorithm~for~Detecting~Faults~in~Rotating~Machines%
\footnote{Preprint submitted to \textit{Mechanical Systems and Signal Processing (Elsevier)}.} %
}

\author[1,2]{Wangpeng He\thanks{Email: wangp.he@gmail.com}}
\author[1]{Yin Ding\thanks{Email: yd372@nyu.edu}}
\author[2]{Yanyang Zi}
\author[1]{Ivan W. Selesnick}

\affil[1]{Tandon School of Engineering, New York University, 6 Metrotech Center, Brooklyn, NY 11201, USA}
\affil[2]{State Key Laboratory for Manufacturing and Systems Engineering, School of Mechanical Engineering, Xi'an Jiaotong University, Xi'an 710049, P. R. China}

\begin{document}

\maketitle

\begin{abstract}
This paper addresses the detection of periodic transients in vibration signals for detecting faults in rotating machines.
For this purpose, we present a method to estimate periodic-group-sparse signals in noise.
The method is based on the formulation of a convex optimization problem.
A fast iterative algorithm is given for its solution.
A simulated signal is formulated to verify the performance of the proposed approach for periodic feature extraction.
The detection performance of comparative methods is compared with that of the proposed approach via RMSE values and receiver operating characteristic (ROC) curves.
Finally, the proposed approach is applied to compound faults diagnosis of motor bearings.
The non-stationary vibration data were acquired from a SpectraQuest's machinery fault simulator.
The processed results show the proposed approach can effectively detect and extract the useful features of bearing outer race and inner race defect.
\end{abstract}

\section{Introduction}

Rotating machinery is one of the most common types of mechanical equipment and plays a significant role in industrial applications.
Early detection of faults developing in rotating machinery is of great importance to prevent economic loss and personal casualties \cite{heng2009rotating}.
Rolling element bearings and gearboxes are two kinds of widely used components in rotating machines and their failures are among the most frequent reasons for machine breakdown.

Much attention has focused on vibration-based diagnosis of mechanical faults in rotating machines \cite{yan2014wavelets}.
The detection of periodically occurring transient vibration signatures is of vital importance for vibration-based condition monitoring and fault detection of rotating machinery \cite{randall2011rolling}.
However, these useful transient features are usually buried in heavy background noise and other irrelevant vibrations.
To address this problem, many signal processing methods have been introduced, such as singular value decomposition (SVD) \cite{cong2013short}, 
empirical mode decomposition (EMD) \cite{lei2013review},
and methods based on different wavelet transforms,
e.g., dual-tree wavelet in \cite{zhang2014periodic}, harmonic wavelet in \cite{yan2010harmonic}, and tunable Q-factor wavelet (TQWT) in \cite{he2014automatic}.
These methods have achieved successful applications in the field of machinery fault diagnosis.

Recently, an algorithm, called `overlapping group shrinkage' (OGS) was developed for estimating group-sparse signals in noise \cite{Chen_Selesnick_2014_OGS}.
The OGS algorithm was initially formulated as a convex optimization promoting group sparsity by a convex regularization.
In order to promote sparsity more strongly, an improved version of OGS was developed, which utilizes a non-convex regularization \cite{Chen_Selesnick_2014_GSSD}.
The superiority of denoising group-sparse signals using the approach presented in \cite{Chen_Selesnick_2014_GSSD} indicates its potential for effectively extracting periodic transient pulses.

This paper aims to develop an approach for rotating machinery fault diagnosis based on a periodic group-sparse signal representation. The signature of localized faults of the gear teeth and bearing components generally exhibit periodic transient pulses when a rotating machine is operated at a constant speed \cite{he2013time}. Meanwhile, the large amplitudes of these useful features are not only sparse but also tend to form groups. Several neighborhood-based denoising methods have been developed for machinery fault diagnosis utilizing this property \cite{chen2012fault, zhen2008customized, he2013tunable}.
Our approach is based on a signal model intended to capture the useful impulsive features for machinery fault diagnosis.
In particular, this paper addresses the problem of estimating $x$ from a noisy observation $y$.
We model the measured discrete-time series, $y$, as
\begin{equation}\label{eq:signal}
	y_n = x_n + w_n, \qquad n = 0, \ldots, N-1,
\end{equation}
where the signal $x$ is known to have a periodic group-sparse property and $w$ is white Gaussian noise.
A group-sparse signal is one where large magnitude signal values tend not to be isolated.
Instead, these large magnitude values tend to form groups.

Convex optimization is commonly used to estimate sparse vectors from noisy signals,
where we solve the optimization problem with the prototype
\begin{equation} \label{eq:optimization}
	x\opt =\arg \min_{x} \Big \{ F(x)= \half \norm{ y - x }_2^2 + \lam \Phi(x) \Big \},
\end{equation}
where
 $\lam >$ 0 is a regularization parameter and $\Phi$ : $\mathbb{R}^N$ $\to$ $\mathbb{R}$ is a sparsity-promoting penalty function (regularizer).
Conventionally, the regularizer $\Phi(x)$ is a convex function, e.g. $\ell_1$-norm.
In \cite{selesnick2014sparse}, an idea of using non-convex regularizer and keeping the convexity of entire problem is used for signal denoising problem, where the sparsity can be significantly promoted comparing to $\ell_1$-norm,
and the problem still has a unique solution due to the convexity.

In this paper, we adopt ideas from \cite{selesnick2014sparse} and \cite{Chen_Selesnick_2014_OGS},
and present a method for estimating periodic-group-sparse signals in noise.
We propose its use for detecting faults in rotating machinery,
where the fault characteristic frequency (period of the group-sparse pulses) is used as prior knowledge.
Similar to the approach in \cite{Chen_Selesnick_2014_GSSD},
the non-convex regularization term in the proposed method is properly chosen so as to ensure that the total objective function $F$ is convex;
however, in contrast to \cite{Chen_Selesnick_2014_GSSD}, where each group has to be contiguous,
we allow grouping with intervals, and furthermore periodically.

As a consequence, in this work, the regularization term $\Phi$ in \eqref{eq:optimization} is formulated specifically to utilize the periodicity of the impulsive fault features.
The aim of our approach is to capture the useful impulsive features for the purpose of machinery fault diagnosis.
Additionally, it has the potential to separate compound fault features by utilizing different periods of the periodic transient pulses corresponding to different fault frequencies (e.g., various defect frequencies of rolling element bearings).
The proposed approach also reduces to a non-periodic group-sparse signal denoising method,
i.e., we can utilize the sparsity-based OGS approach \cite{Chen_Selesnick_2014_GSSD}, if we do not have prior knowledge of the period of the group-sparse transients.
Thus, the proposed sparsity-based approach is a generalization of the non-convex regularized OGS \cite{Chen_Selesnick_2014_GSSD}.

The paper is organized as follows. A brief review of OGS with convex and non-convex regularization is given in Section 2. Section 3 presents a method for denoising periodic group sparse signals. In Section 4 a simulation study is performed to validate the effectiveness of the proposed method. Section 5 applies the proposed periodic group sparse denoising method to fault diagnosis of motor bearings for further validation of its effectiveness. Finally, conclusions are summarized in Section 6.

\section{Review}\label{sec:review_OGS}

In this section, we give short reviews of overlapping group shrinkage (OGS) \cite{Chen_Selesnick_2014_OGS}
and majorization-minimization (MM) \cite{opt_Lange_book_2004}.

\subsection{Overlapping Group Shrinkage (OGS)}
There are several advantages to formulating sparse estimation as a convex optimization problem.
The most basic advantage is that the problem can then be very reliably and efficiently solved using convex optimization methods \cite{boyd2009convex}.
Although a non-convex regularizer can promote sparsity more strongly, it generally leads to a non-convex
optimization problem with non-optimal local minima \cite{selesnick2015convex}.
To avoid the formulation of a non-convex optimization problem, one may utilize a non-convex regularizer $\Phi$ designed so as to ensure the total objective function is convex.

The problem of denoising a group sparse signal was addressed in \cite{Chen_Selesnick_2014_OGS} which utilized convex optimization. 
An improved method was proposed in \cite{Chen_Selesnick_2014_GSSD}, 
which utilized non-convex regularization designed to ensure convexity of the objective function. 
The problem is solved efficiently by an iterative algorithm based on majorization-minimization (MM) \cite{opt_Lange_book_2004}. 
The objective function for the OGS problem, with a group size of $K$ and convex or non-convex regularization, is denoted as
\begin{align}\label{eqn:fd_cost_function_GSSD}
	x \opt = \arg \min_{x \in \real^N} \Bigg\{ P_0(x) 
	= \half \norm{y-x}_2^2 + \lam \sum_{n}  \phi \bigg( \Big[\sum_{k \in \Kset} x_{n+k}^2 \Big]^{1/2} ;a \bigg) \Bigg\}
\end{align}
where $\Kset : = \{ 0, 1, \ldots K-1 \} $,
and $y \in \real^{N}$ is the noisy observation. For $x \in \real^N$, we define $x_n=0$ for $n < 0$ and $n \ge N$.
We assume the penalty function $\phi : \real \to \real_{+}$ satisfies the properties:
\begin{enumerate}
	\item
		$ \phi $ is continuous on $ \real $.
	\item
		$ \phi $ is twice differentiable on $\real \setminus \{ 0 \}$.
	\item
		$ \phi $ is increasing and concave on $ \real_+ $.
	\item
		$ \phi $ is symmetric, $ \phi(-x; a) = \phi(x; a) $.
	\item
		$ \phi'(0^+; a) = 1 $.
	\item
		$ \phi''(x; a) \ge -a $ for all $ x \neq 0 $.
	\item
		$ \phi(x; 0) = |x| $.
\end{enumerate}
which is used to induce the resulting sparsity in an optimization problem.
\begin{figure}[t]
\centering
\includegraphics[scale = 0.8]{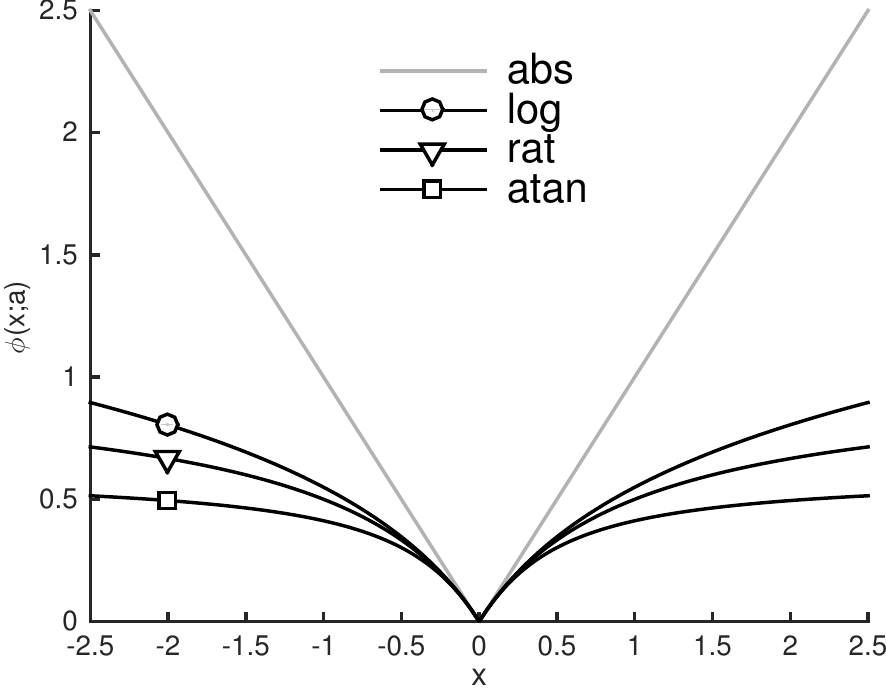}
\caption{Examples of penalty functions.}
\label{fig:example_pen}
\end{figure}

\begin{table}[t]
	\caption{Sparsity-promoting penalty functions $(a \ge 0)$.} 
	\label{tab:popu_penalty}
\begin{center}
	\scalebox{1.00}{
		\begin{tabular}{@{} lll @{}}
		\bottomrule
		\textbf{Penalty}	 & $\phi(x ; a) $   & $ \psi(x) $ \\[0.4em]
		\midrule
			abs	 	& $ \abs{x} $						
	 				& $ \abs{x} $ \\[0.8em]	
			log	 	& $ \ds \frac{1}{a} \log ( 1 + a \abs{x} )$
	 				& $ \ds \abs{x} ( 1 + a \abs{x} ) $	\\[0.8em]
			rat 	& $ \ds  \frac{\abs{x}}{1+a\abs{x}/2}$
					& $ \ds \abs{x} ( 1 + a \abs{x} )^2 $	\\[0.8em]
			atan 	&$\ds\frac{2}{a \sqrt{3}} \left( \tan^{-1} \left( \frac{1+2 a \abs{x}} {\sqrt{3}} \right) - \frac{\pi}{6}\right)$
	 				& $ \ds \abs{x} ( 1 + a \abs{x} + a^2 \abs{x}^2 ) $	\\[0.8em]
		\toprule
		\end{tabular}
	}
	\end{center}
\end{table}

Note that the objective function $P_0$ in \eqref{eqn:fd_cost_function_GSSD} may be convex even if the regularizer (penalty) is not.
When the penalty function $\phi$ satisfies the conditions listed above, 
then the parameter `$a$' can be chosen so that $P_0$ is convex even if $\phi$ is not \cite{selesnick2014sparse}.
This approach is used to improve OGS in \cite{Chen_Selesnick_2014_GSSD} 
where it has been proved that the objective function $P_0$ in \eqref{eqn:fd_cost_function_GSSD} is strictly convex if
\begin{align}
	0 \le a < \frac{1}{\lam K}.
\end{align}

Table~\ref{tab:popu_penalty} gives some examples of penalty functions satisfying the above-listed conditions, 
and these examples are illustrated in \figref{fig:example_pen}.
In order to induce the sparsity most strongly,
the arctangent `atan' function can be used among all the three given non-convex functions.
\figref{fig:example_pen} also illustrates that under the same `$a$' parameter,  `atan' function adheres the most `concavity'.
In addition to these non-convex penalty functions, 
\figref{fig:example_pen} also shows the $\ell_1$-norm as a special case.
Note that we allow $ a = 0 $, which is the extreme case of the three other penalties.
In particular, $ \phi(x;0) = \abs{x}$ is a convex function.
The original OGS algorithm given in \cite{Chen_Selesnick_2014_OGS} considers only the special case $a=0$, i.e., convex regularization.

\subsection{Majorization-minimization (MM)}
A detailed derivation to solve problem \eqref{eqn:fd_cost_function_GSSD},
based on the majorization-minimization (MM) method is given in \cite{Chen_Selesnick_2014_GSSD}.
The MM method simplifies a complicated optimization problem into a sequence of easy ones, and is described by the iteration
\begin{align}\label{eqn:fd_mm_iteration}
	u^{(i+1)} = \arg \min_{u} G( u , u^{(i)} ),
\end{align}
where $i$ denotes the number of iterations,
and $G : \real^{N} \times \real^{N} \to \real$ is a majorizer of the objective function $J$, satisfying
\begin{subequations}
\begin{align}
	&G(u,v) \ge J(u),\quad \text{ for all } u, v,\\
	&G(u,u) = J(u).
\end{align}
\end{subequations}

A majorizer of $\phi$ is given by
\begin{align} \label{eqn:fd_guv}
	g(u,v) = \frac{(u^2 - v^2)}{2 \psi(v)} + \phi(v;a), \quad  \text{ when } v \neq 0,
\end{align}
where $ \psi $ is the function given in Table~\ref{tab:popu_penalty}.
For $ v = 0 $, $g(u,0)$ is defined by
\begin{align}\label{eqn:fd_singularity}
	g(u,0) : = \begin{cases}
			+\infty,	& \text{ if }u \neq 0, \\
			 0,			& \text{ if }u = 0.
	\end{cases}
\end{align}

As a consequence, the function $g : \real \times \real \to \real_{+} $ satisfies
\begin{subequations}
\begin{align}
	g(u,v) & \ge \phi(u;a), \quad \text{ when } u \neq v, \\
	g(u,u) &  = \phi(u;a),
\end{align}
\end{subequations}
Note that $g(u,0)$ defined in \eqref{eqn:fd_singularity} equals infinity except $u=0$.
This forces its minimizer to lock to $u=0$ in the MM iteration described in \eqref{eqn:fd_mm_iteration}.
This issue in the OGS problem is discussed in \cite{Chen_Selesnick_2014_GSSD},
where it does not affect the convergence when the algorithm is implemented with a non-zero initialization.

\section{OGS with binary weights}

To facilitate the following derivation, we define a binary sequence $b = [b_0, b_1, \ldots b_{K-1}]$, with $b_k \in \{ 0, 1\}$,
and sets
\begin{subequations}\label{eqn:fd_Ksets}
\begin{align}
	&\Kset		:=  \{ 0, 1, \ldots , K-1 \}, \label{eqn:fd_K} \\
	&\Kset_0 	:= \{ k \in \Kset: b_k = 0 \}, \label{eqn:fd_K0}\\
	&\Kset_1 	:= \{ k \in \Kset: b_k = 1 \}. \label{eqn:fd_K1}
\end{align}
\end{subequations}
%
Since $b$ is a binary vector, we have $\Kset = \Kset_0 \cup \Kset_1$ and $ \Kset_0 \cap \Kset_1 = \emptyset$ is the empty set.
We denote the cardinality (size) of the sets $\Kset$, $\Kset_1$ and $\Kset_0$ as
$K$, $K_1$, and $K_0$, respectively, so that $K = K_0 + K_1$, and
\begin{align} \label{eqn:definK1}
	\sum_{ k \in \Kset } b_k = \sum_{ k \in \Kset_1 } b_k & = K_1.
\end{align}

The following proposition is straightforward.

\newtheorem{proposition}{Proposition}
\begin{proposition}\label{prop:prop_1}
Let $\phi : \real \to \real$ satisfy the conditions listed in Section~\ref{sec:review_OGS}. When $\gamma > 0$ and $\lam > 0$, the function $p : \real \to \real$,
\begin{align}\label{eqn:fd_pv}
	p(v) = \frac{\gamma}{2} v^{2} + \lam \phi( v \; a )
\end{align}
is strictly convex if
\begin{align} \label{eqn:fd_phi_condition}
	\phi''(v;a) > -\frac{\gamma}{\lam}.
\end{align}
\end{proposition}

Detailed proof of this proposition can be found in Appendix~\ref{app:A}.

\subsection{Problem definition}
We define the objective function $P_1 : \real^{N} \to \real$ as
\begin{align}\label{eqn:fd_pogs}
	x \opt = \arg \min_{x} \Bigg\{ P_1(x) = \half \norm{y-x}_2^2 + \lam \sum_{n}  \phi \big( \bnorm(x,b,n) ;a \big) \Bigg\}
\end{align}
where the binary-weighted grouping function $\bnorm : \real^{N} \times \real^{K} \times \mathbb{Z}  \to \real $ is defined as
\begin{align}\label{eqn:fd_bnorm}
	\bnorm (x,b,n) : = \Big[  \sum_{k=0}^{K-1} b_k x_{n+k}^2 \Big]^{1/2},
\end{align}
which is the Euclidean norm of a binary weighted block.
For $x \in \real^N$, we define $B_n(x) \in \real^{K}$ as
\begin{align}\label{eqn:fd_Bn}
	B_n(x) := [x_n, x_{n+1}, \ldots x_{n+K-1}],
\end{align}
i.e., a $K$-point subvector of $x$, starting at index $n$. Hence $\bnorm (x,b,n)$ can be written as
\begin{align}
	\bnorm (x,b,n) = \norm{ b \odot B_n(x) }_2,
\end{align}
where $\odot$ denotes element-wise multiplication.

Note that in \eqref{eqn:fd_pogs} if $b_k = 1$ for all $k \in \Kset$, then $ \bnorm^2 (x,b,n) = \sum_k x_{n+k}^2$,
and problem \eqref{eqn:fd_pogs} reduces to \eqref{eqn:fd_cost_function_GSSD}.
Therefore \eqref{eqn:fd_cost_function_GSSD} is a special case of \eqref{eqn:fd_pogs}.
Moreover, if $K = K_1 = 1$, this problem further reduces to scalar (i.e., non-group) sparse denoising.
In the following discussion, we consider the group sparse case: $N \gg K \ge K_1 > 0$. The case $K_1 = 0$ is trivial.

In the following section, we exploit the convexity condition of \eqref{eqn:fd_pogs},
to constrain the non-convex penalty function $\phi$, so as to ensure that the objective function $P_1$ is convex.
Therefore, the problem formulation we ultimately propose is a convex optimization problem. Hence,
it will have no non-optimal local minima in which an iterative optimization algorithm can be trapped.

\subsection{Convexity conditions}

\begin{proposition}\label{prop:prop_2}
Let $b \in \{ 0, 1\}^K$ be a binary vector with $ \sum_k b_k = K_1$. The function $P : \real^{K} \to \real$ defined as
\begin{align}\label{eqn:fd_ogs_Pu}
	P(u) = \frac{1}{2 K_1} \sum_{k=0}^{K-1} b_k u_k^2 + \lam \phi\bigg( \Big[\sum_{k=0}^{K-1} b_k u_{k}^2 \Big]^{1/2} ;a \bigg),
\end{align}
is strictly convex if
\begin{align}\label{eqn:fd_ogs_prop_2}
	\phi''(x;a) > - \frac{1}{ K_1 \lam} \qquad  \text{ for all } x \neq 0.
\end{align}
\end{proposition}

This Proposition is proved in detail in Appendix~\ref{app:B}.
Utilizing the above results, we find a range of $a$ ensuring that,
even though $\phi$ is non-convex, the objective function $P_1$ in \eqref{eqn:fd_pogs} is strictly convex.

\newtheorem{theorem}{Theorem}
\begin{theorem}\label{the:the_1}
Assume that problem $P_1$ \eqref{eqn:fd_pogs} has $K_1 = \sum_k b_k$ non-zero weights in one binary weighted group,
then the objective function is strictly convex
if the parameter `$a$' of the penalty function $\phi(\cdot;a)$ satisfies
\begin{align}\label{eqn:fd_ogs_the_a}
	0 \le a < \frac{1}{ K_1 \lam }.
\end{align}
\end{theorem}

\begin{proof}
First, since $  \sum_k b_k = K_1 $ and  $b_k\ge 0$, $x_n^2 \ge 0$, it follows that
\begin{align} \label{eqn:fd_prf_theorem_1}
	K_1 \sum_n x_{n}^2 	& = \sum_k b_k \sum_n x_{n}^2 =  \sum_n \sum_k  b_k  x_{n}^2.
\end{align}
Using \eqref{eqn:fd_prf_theorem_1}, we write $\half \sum_n x_n^2 = \frac{1}{2K_1}\sum_{n} \sum_k b_k x_{n}^2$.
Therefore, the data fidelity term in problem \eqref{eqn:fd_pogs} can be written as
\begin{align}
	F(x) 	& = \half \norm{y-x}_2^2 = \half \sum_{n} x_n^2 +  L(x) \nonumber \\
			& = \frac{1}{2 K_1} \sum_{n=0}^{N-K} \Big(  \sum_k b_k x_{n+k}^2 \Big) + L(x),		
\end{align}
where $L(x)$ is linear in $x$.
Adding $ L(x) $ to a strictly convex function yields a strictly convex function.

Using the above results, the objective function in problem \eqref{eqn:fd_pogs} can be reorganized as
\begin{align} \label{eqn:fd_theorem}
	P_1(x)	& = \sum_{n=0}^{N-K} \Bigg[ \frac{1}{2 K_1} \sum_{ k } b_k x_{n+k}^2 + \lam \phi\bigg( \Big[\sum_{ k } b_k x_{n+k}^2 \Big]^{1/2} ;a \bigg) \Bigg] + L(x), \nonumber \\
			& = \sum_{n=0}^{N-K} P( B_n(x) )+ L(x),
\end{align}
where $B_n(x)$ is defined in \eqref{eqn:fd_Bn}.
As a consequence, if $ P $ is strictly convex, then $ P_1 $ is strictly convex.
The condition for convexity of $P$ is given in Proposition~\ref{prop:prop_2}.
Hence, as long as the inequality condition \eqref{eqn:fd_ogs_prop_2} is satisfied, $P$ is strictly convex.
Moreover, $\phi$ satisfies condition $ \phi''(x; a) \ge -a $ (condition~6 in Section~\ref{sec:review_OGS}).
This implies that when \eqref{eqn:fd_ogs_the_a} is satisfied, $ P $ is strictly convex, and the entire objective function $P_1$ is convex.
\end{proof}

Note that Theorem~\ref{the:the_1} generalizes the convexity condition of OGS in \cite[Corollary 2]{Chen_Selesnick_2014_GSSD}.
When every element in binary vector $b$ equals $1$ then $K_1 = K$ and Theorem~\ref{the:the_1} reduces to Corollary 2 in \cite{Chen_Selesnick_2014_GSSD} as a special case.

We have proved that under a more flexible group structure (binary weights),
non-convex penalty functions can be utilized to promote structured sparsity,
and the convexity of the objective function is preserved when the regularization parameter $a$ is suitably set.
Moreover, the result also shows that
when maximizing the non-convexity of the penalty function,
only the the nonzero weights matter for the selection of the parameter $a$.

\subsection{Algorithm derivation}

To minimize $P_1$ using the MM procedure, we define a majorizer $G: \real ^N \times \real ^N \to \real$, defined by
\begin{align}
	G(x, v) & = \half \norm{ y - x }_2^2 + \lam \sum_{n} g\big( \bnorm(x,b,n), \bnorm(v,b,n)  \big).
\end{align}
We can write $G$ as
\begin{subequations}
\begin{align}
	G(x, v) & = \half \norm{ y - x }_2^2 + \frac{\lam}{2} \sum_{n} \frac{1}{\psi( \bnorm(v,b,n)  )}\bnorm^2(x,b,n) + C 		\\
			& = \half \norm{ y - x }_2^2 + \frac{\lam}{2} \sum_{n} \sum_k\frac{ b_k }{\psi( \bnorm(v,b,n)  )}x^2_{n+k} + C 	\\
			& = \half \norm{ y - x }_2^2 + \frac{\lam}{2} \sum_n [r(v)]_n x_n^2 +C
\end{align}
\end{subequations}
where $r(v) \in \real^N$ is defined as
\begin{align}\label{eqn:fd_rn}
	[r(v)]_n : =
			\displaystyle \sum_{j=0}^{K-1} \frac{ b_j } {\psi\big(  \bnorm( v, b, n-j ) \big)}.
\end{align}

Then $G(x,v)$ can be written as
\begin{align}
	G(x, v) & = \half \sum_{n} x_n^2  + \frac{\lam}{2} \sum_n r_n x_n^2 - \sum_n y_n x_n + C(y) \\
			& = \sum_{n} \Big( \half + \frac{ \lam [r(v)]_n }{2}\Big) x_n^2 - y_n x_n +C(y)
\end{align}
which has an explicit minimizer $ x_n = {y_n} /( { 1+ \lam [r(v)]_n } ) $.
Hence, the MM iteration \eqref{eqn:fd_mm_iteration} is given by
\begin{align}\label{eqn:fd_result}
	x_n^{(i+1)} = \frac{y_n} { 1+ \lam [r(x_n^{(i)})]_n }.
\end{align}
Table~\ref{alg:ogs_b} gives the explicit steps to solve \eqref{eqn:fd_pogs},
assuming the penalty function is chosen from Table ~\ref{tab:popu_penalty} and satisfies \eqref{eqn:fd_ogs_the_a}.
This guarantees the problem \eqref{eqn:fd_pogs} is strictly convex and consequently MM procedure \eqref{eqn:fd_mm_iteration} will converge to the unique global minimizer.

Note that zero-locking might occur when using quadratic function to majorize non-smooth function \cite{FBDN_2007_TIP}.
For the OGS problem, initializing the algorithm by $x^{(0)} = y$ avoids this issue; a detailed proof is given in \cite[Lemma~B]{Chen_Selesnick_2014_GSSD}.
This lemma is not affected by introducing binary group weights.
In other words, if the function $F$ in \cite[Lemma~B]{Chen_Selesnick_2014_GSSD} is substituted  with $P_1$ \eqref{eqn:fd_pogs},
the derivation is still valid with an almost identical proof.
Moreover, since we allow $\psi(x)$ to be 0, Equation \eqref{eqn:fd_rn} may lead to a `divide-by-zero'.
The work of \cite{Chen_Selesnick_2014_GSSD} contains a sufficient discussion that this problem is avoided by the initialization $x^{(0)} = y$,
based on a same lemma. 
Hence, there is no zero-locking or `divide-by-zero' issue when solving the problem \eqref{eqn:fd_pogs} by the algorithm in Table~\ref{alg:ogs_b}.

\begin{table}[htbp]
\caption{OGS with binary weights.}
\label{alg:ogs_b}
\begin{subequations}
	\begin{empheq}[box=\fbox]{align*}
		& \text{Input:} ~y \in \real^N, ~\lam, ~b\in \{ 0, 1 \}^{K}  		\nonumber		\\
		& \text{Initialization:} ~x = y, ~ S = \{ n : y_n \neq 0\} \\
		& \text{Repeat for $ n \in S$:}  \\
		& \qquad r_n = \sum_{j=0}^{K-1} \frac{ b_j } { \psi\big(  \bnorm( x, b, n-j ) \big) } \\
		& \qquad x_n =  \frac{y_n} { 1+ \lam r_n } \\
		& \qquad S = \{ n: \abs{x_n} > \eps \} 	  \\
		& \text{Until convergence} \\
		& \text{Return: } x
	\end{empheq}
\end{subequations}
\end{table}

\subsection{Periodicity-induced OGS (POGS)}

In the previous sections, we have given a method for group sparse denoising with binary weights within the group.
In signal model \eqref{eq:signal}, since the periodicity of impulsive faults in $x$ is assumed to be approximately consistent over a reasonable duration, the time interval between two consecutive faults can be considered identical within the support of a group.
Moreover, when the period $T$ of a potential fault is known or predictable from the knowledge of the machinery,
we can select the group with a length $K$
and its zero and non-zero entries by
\begin{subequations}\label{eqn:fd_preoid_K}
\begin{align}
 	 N_0 + N_1  &\approx f_s T, \\
	 N_0 + N_1  &= K /M,
\end{align}
\end{subequations}
where $N_1$ and $N_0$ are the estimated length (in samples) of impulsive transients and the time interval between them,
and integer $ M \ge 2 $ is the number of periods contained in one group, and $f_s$ is sampling rate.
Thus, in one group, the numbers of zero and non-zero entries are $K_0 = M N_0$ and $K_1 = M N_1$, respectively.

Moreover, when the transient sequence is periodical,
the binary weight $b_k$ are chosen according to a periodic group structure. Specifically, under the constraint \eqref{eqn:fd_preoid_K},
$b \in \{ 0,1 \}^K$ is given by
\begin{align}\label{eqn:fd_preoid_b}
		b = [
			\ 	\underbrace{ 1, 1, \ldots , 1 }_{N_1},
			\	\underbrace{ 0, 0, \ldots , 0 }_{N_0},
			\	\underbrace{ 1, 1, \ldots , 1 }_{N_1},
			\ 	\underbrace{ 0, 0, \ldots , 0 }_{N_0} ,
			\	\ldots, \	
			\ 	\underbrace{ 1, 1, \ldots , 1 }_{N_1},
			\ 	\underbrace{ 0, 0, \ldots , 0 }_{N_0} \
			],
\end{align}
where in each period, there are $N_1$ non-zero entries grouped according to the impulsive signal,
and the entire group comprises $M$ periods.
In this case, the last $N_0$ zeros in \eqref{eqn:fd_preoid_b} has no effect in problem \eqref{eqn:fd_pogs} by the definition \eqref{eqn:fd_bnorm}.
In practice, we omit trailing zeroes and the actual length of $b$ involved in computation is $K-N_0$.
Note that although we use the parameters such as $K, K_0, K_1, N_0, N_1$ and $M$ to illustrate the binary pattern structure,
in fact given $f_s$ and $T$, we only need to select $N_1$ and $M$,
then use \eqref{eqn:fd_preoid_K}, the pattern of $b$ \eqref{eqn:fd_preoid_b} can be determined.

Consequently, we propose to recover a periodic impulsive signature from a noisy observation by solving problem \eqref{eqn:fd_pogs} with $b$ defined by \eqref{eqn:fd_preoid_K} and \eqref{eqn:fd_preoid_b}.
We refer to this method as Periodicity-induced OGS (POGS). This is an extension of OGS accounting for the periodicity of the sparse signal.

As a special case, if the period $T$ is not known, we may use conventional OGS \eqref{eqn:fd_cost_function_GSSD} to detect faults.
If a period $T$ can be determined by inspecting the output produced by OGS,
then an enhanced result with better accuracy may be achieved by POGS using the determined periodicity.

\section{Simulation validation}

\begin{figure}[t]
	\centering
    \includegraphics[scale = \figurescale]{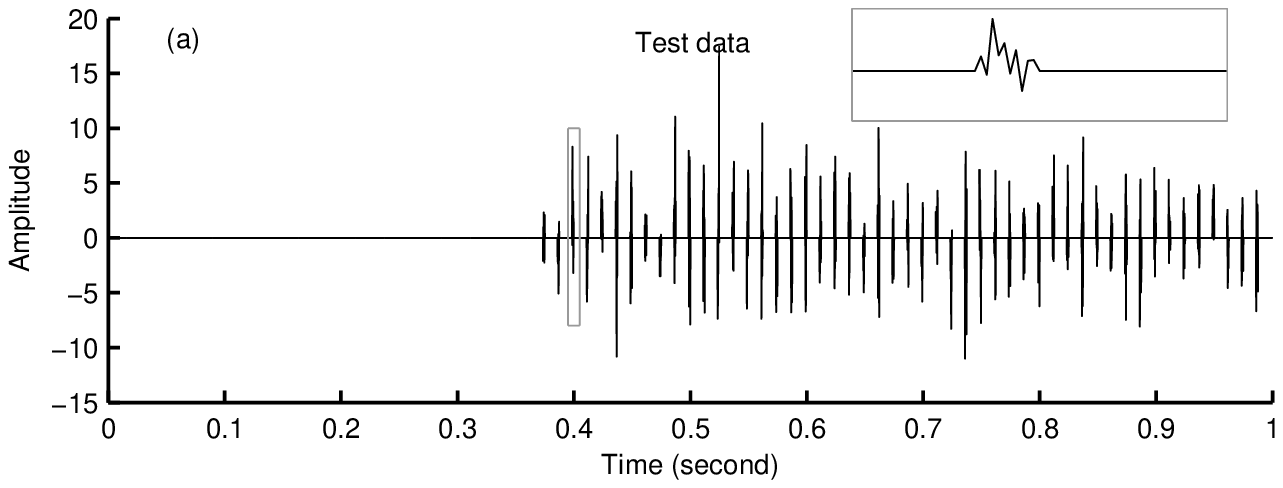}
    \\
    \includegraphics[scale = \figurescale]{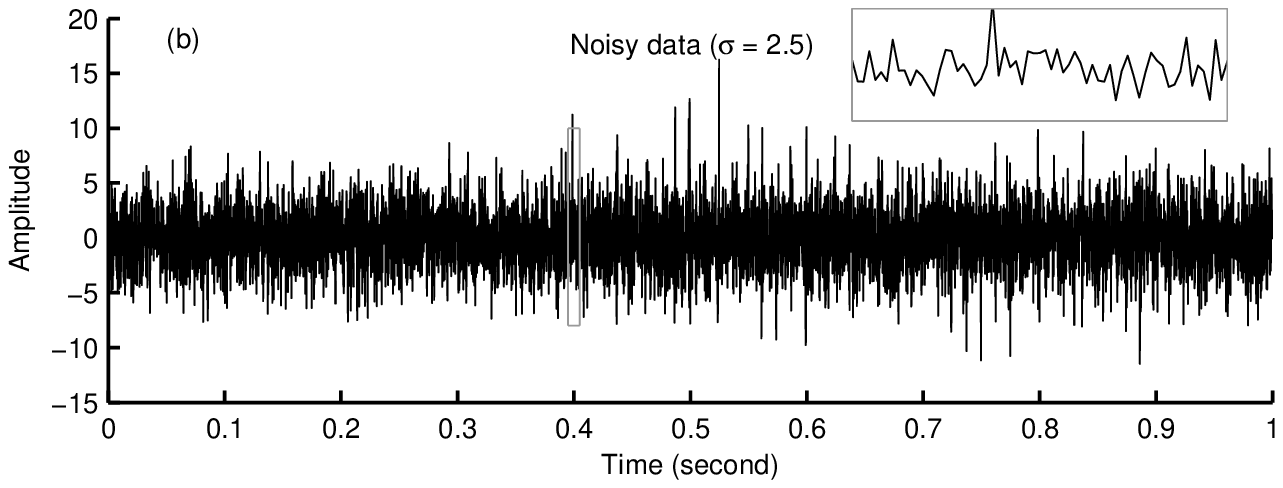}
	\caption{ Example 1: Simulated signal: (a) clean data and (b) noisy data.}
	\label{fig:example1_data}
\end{figure}

\begin{figure}[!t]
	\centering
    \includegraphics[scale = \figurescale]{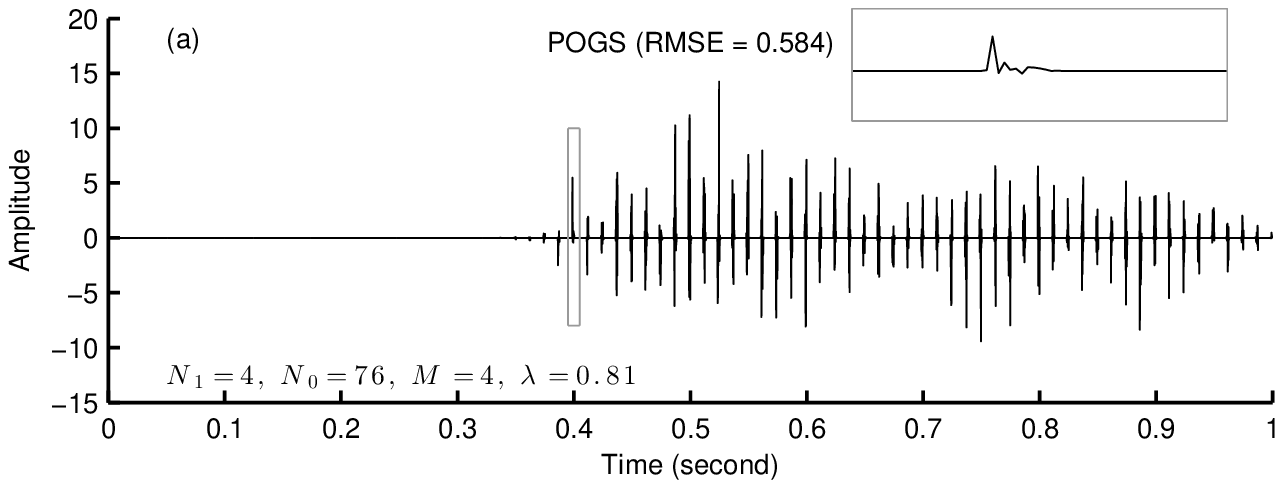}
    \\
    \includegraphics[scale = \figurescale]{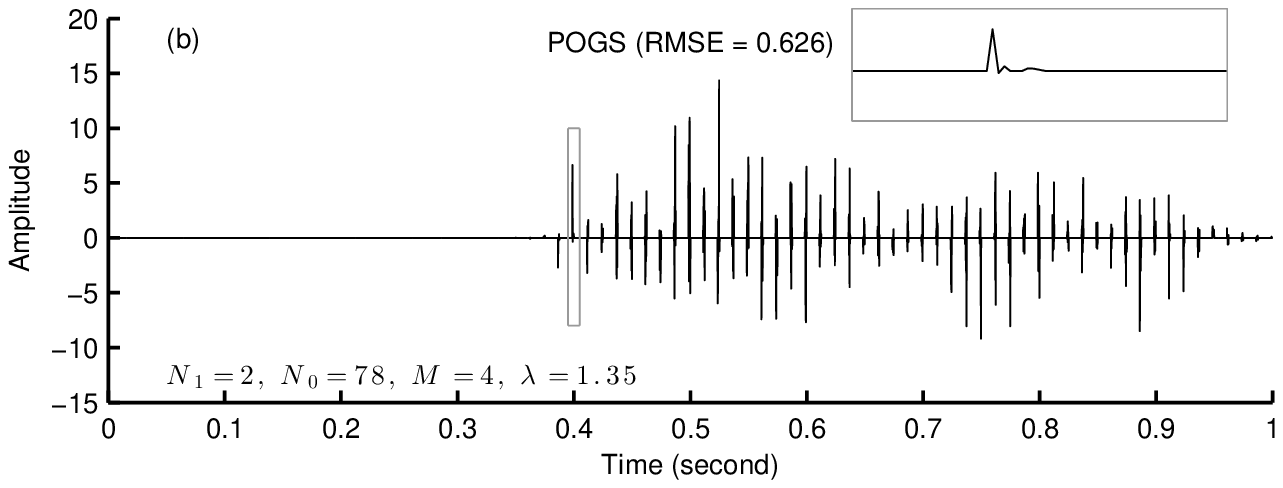}
    \\
	\includegraphics[scale = \figurescale]{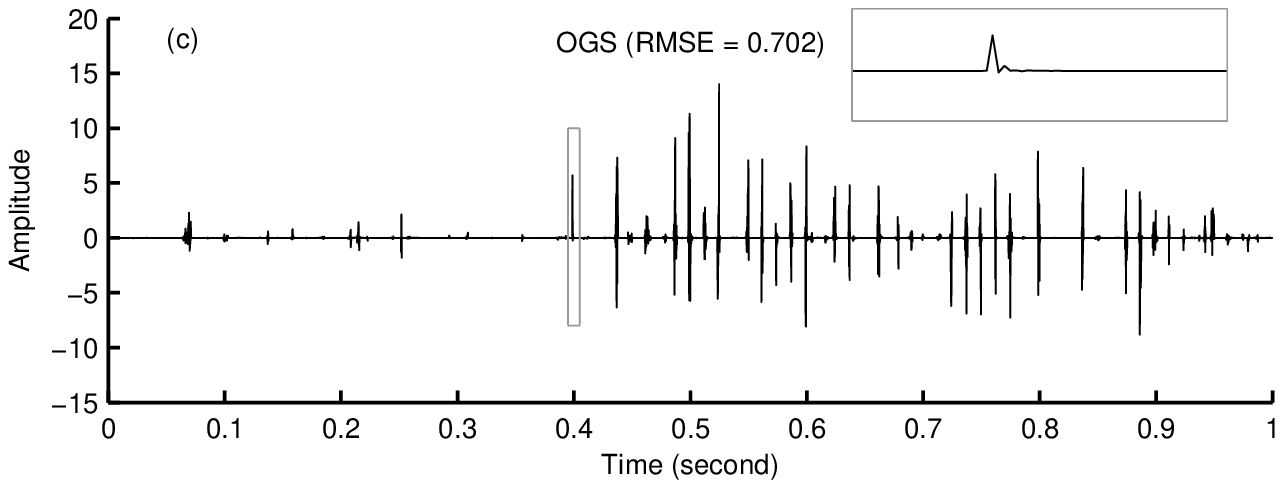}
	\\
	\includegraphics[scale = \figurescale]{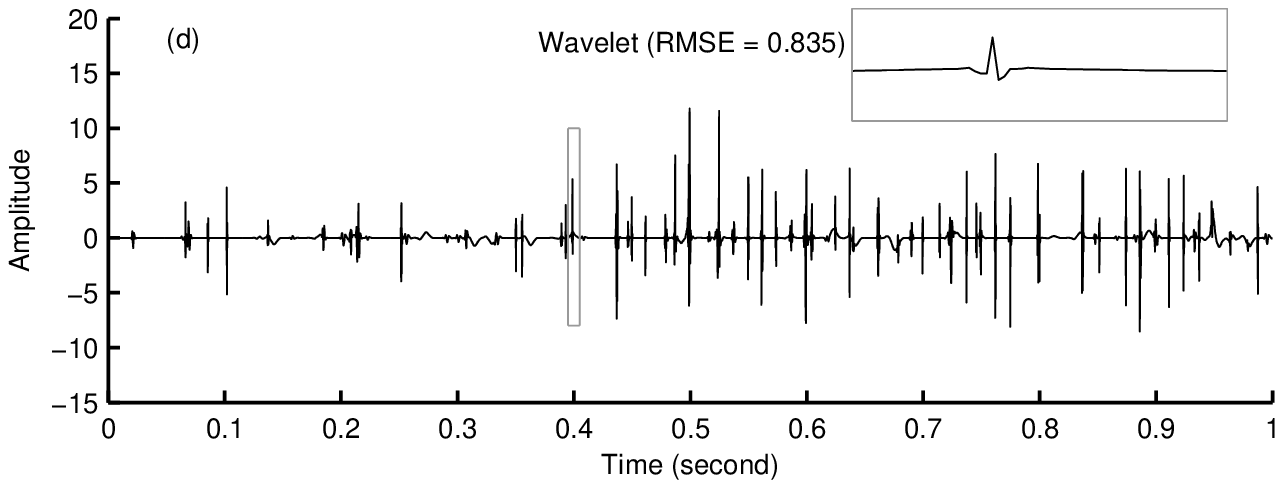}
	\caption{ Example 1: Denoising results.
		(a) Proposed method with $N_1 = 4, M = 4$.
		(b) Proposed method with $N_1 = 2, M = 4$.
		(c) Conventional OGS.
		(d) Wavelet-based denoising.}
	\label{fig:example_1_results}
\end{figure}

To test the proposed method, we apply it to the simulated data illustrated in \figref{fig:example1_data}(a).
The signal is a 1-second signal with sampling rate $ f_s = 6400 $ Hz, and is composed of a periodic sequence of transients occurring with $80$ Hz.

We simulate the vibration signal containing features caused by machinery defect
as a sequence of impulsive transients,
and each transient consists for 10 samples (1.6 ms when $f_s = 6400$ Hz).
In this example, each transient is composed of a random number of sinusoidal components, each with random frequency and random initial phase.
More specifically, each transient can be written as
\begin{align}\label{eqn:fd_transient}
	v(n) = \sum_{i=1}^{U} A_i \sin( \omega_i n + \beta_i ), \quad n = 0,1,2 \ldots 9,
\end{align}
where $1 \le U \le 10$ is a random integer, and for each $i$,
$A_i$ is a random amplitude,
and $\omega_i $ is a random frequency,
and $\beta_i $ is a random initial phase.
The sequence of transients is shown in \figref{fig:example1_data}(a),
and we show the detail of one transient at about $t = 0.40$ second in the box.
The generated test signal is multiplied by a constant,
so that it has unit standard deviation.

In order to evaluate the false alarm rate, the first part of the test signal contains no transient.
\figref{fig:example1_data}(a) shows the clean test signal, where there are 50 faults starting at approximately $t = 0.36$ second with a period $T = 1/80$ second.
White Gaussian noise with standard deviation $ \sigma = 2.5 $ is added to the simulated fault sequence, as illustrated in \figref{fig:example1_data}(b).

\begin{figure}[t]
\centering
	\includegraphics[scale = 0.2]{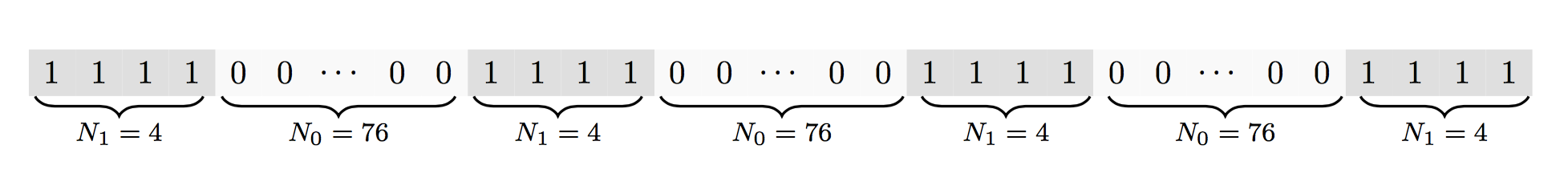}
\caption{Example~1: Binary pattern $b$ when $M = 4$, $N_1 = 4$.}
\label{fig:fd_example_1_b}
\end{figure}

When the periodicity of the faults is known, we can use POGS, where
\eqref{eqn:fd_preoid_K} and \eqref{eqn:fd_preoid_b} are used to define the binary weight vector $b$.
In this example, we set $ N_1 = 4 $ and $ M = 4 $. Since $ T = 1/80 $ seconds, we can calculate $ N_0 = 6400/80 - N_1 = 80 - 4 = 76 $ according to conditions \eqref{eqn:fd_preoid_K},
then the explicit pattern of $b$ can be determined as \figref{fig:fd_example_1_b}.
The denoising result is shown in \figref{fig:example_1_results}(a),
and the transients can be easily identified with an almost pure zero baseline.
In this example, we use OGS and POGS with the atan penalty function and non-convexity parameter $a$ set to its maximum value of $1/(K_{1} \lam )$,
so as to maximally induce sparsity subject to the constraint that the objective function $P_1$ is convex.

\figref{fig:example_1_results}(b) shows another example using $N_1 = 2$ to determine the pattern $b$.
The result is slightly worse in RMSE than \figref{fig:example_1_results}(a),
because $N_1 =2 $ in this example does not match the simulated data.
However, $ N_1 = 2 $ is the lower limit of a realistic value,
because in practice it is very rare that the transients are all single-sample spikes, when the data is properly measured.
As a consequence, the result in \figref{fig:example_1_results}(b) can be understand as the worst case of choosing an inappropriate $b$.

\figref{fig:example_1_results}(c) shows the result of denoising using conventional non-periodic OGS \eqref{eqn:fd_cost_function_GSSD}
with group size $ K = 8 $ and the arctangent (atan) non-convex penalty function.
The regularization parameter $\lam$ is chosen by the look-up table in \cite{Chen_Selesnick_2014_OGS},
which sets $\lam$ proportional to the noise $\sigma$.
In this example, we set $ \lam = 0.52 \sigma $.
The result in \figref{fig:example_1_results}(c) misses some faults and yields several false detections,
e.g., the ones at about $0.82$ second, and some false transients appear before $t = 0.36$ second.

As a comparison, we adopt conventional wavelet-based denoising method to the test signal.
More specifically, a 5-scale undecimated wavelet transform \cite{Coifman_1995} with 3 vanishing moments is used,
and the result is shown in \figref{fig:example_1_results}(d).
For denoising, we apply hard-thresholding and chose the threshold value by $3\sigma$-rule for each subband.
In \figref{fig:example_1_results}(d), although some large amplitude transients can be recovered at correct locations,
they exhibit the same shape as the chosen wavelet
[see the box in \figref{fig:example_1_results}(d)].

\begin{figure}[!t]
	\centering
	\quad
    \includegraphics[scale = \figurescale]{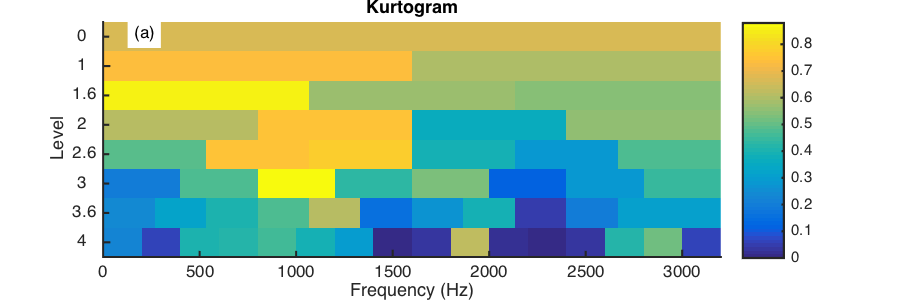}
	\\
    \includegraphics[scale = \figurescale]{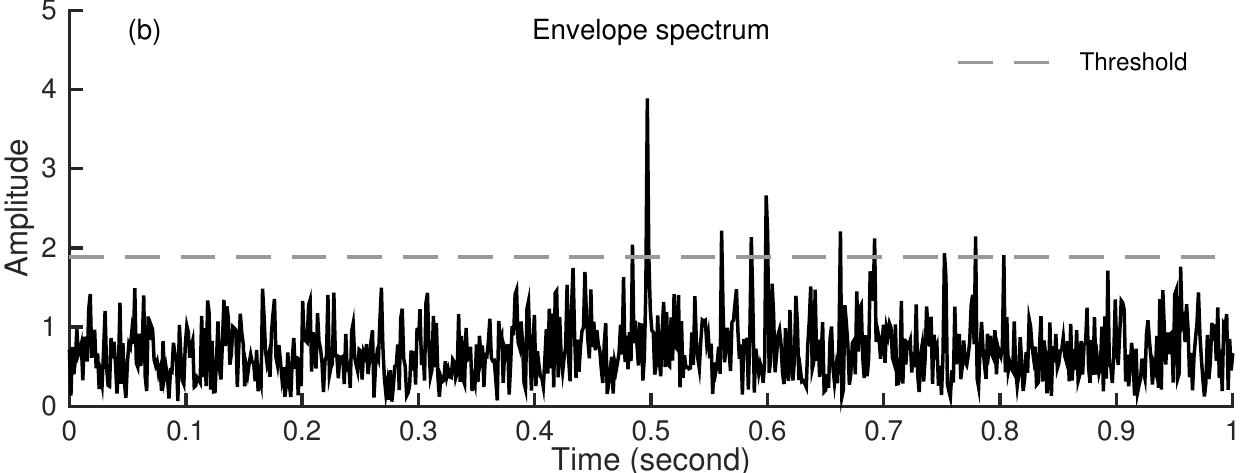}
	\caption{ Example 1: Results of fast spectral kurtosis.
	(a) Kurtogram. (b) Envelope of extracted transients.}
	\label{fig:example1_kurt}
\end{figure}

We compare the performance with fast spectral kurtosis \cite{antoni2006spectral}
%
This method produces the kurtogram in \figref{fig:example1_kurt}(a),
where the kurtosis maximum is at the third level with an estimated  `optimal carrier frequency' at 1000 Hz [see the bright yellow area in \figref{fig:example1_kurt}(a)].
The corresponding amplitude of the extracted transients is shown in \figref{fig:example1_kurt}(b),
with an automatically generated threshold shown as a gray dashed line.
The peaks after $ t = 0.36 $ seconds have a greater density, which indicates that it is more likely that faults occur after $t = 0.36$ seconds.
However, the useful repetitive transients are surrounded by strong irrelevant noise.

\subsection{Parameter selection}

\begin{figure}[!t]
	\centering
    \includegraphics[scale = \figurescale]{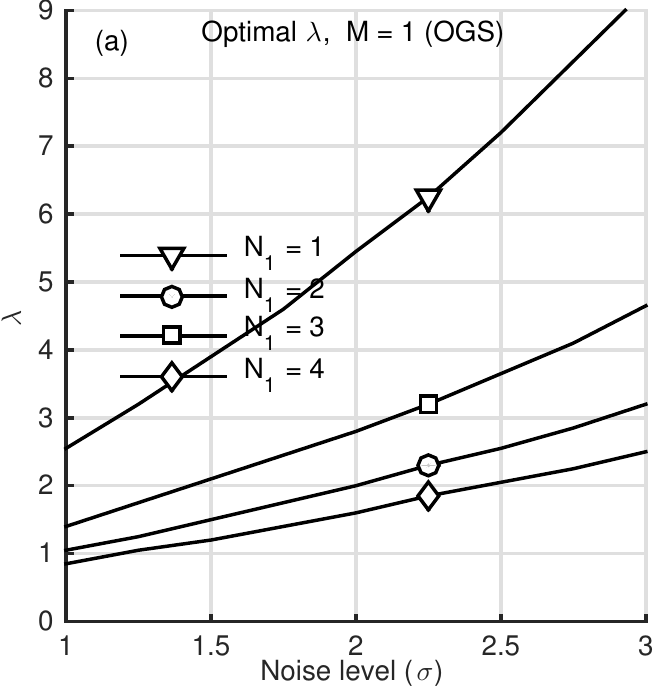}
    \quad
    \includegraphics[scale = \figurescale]{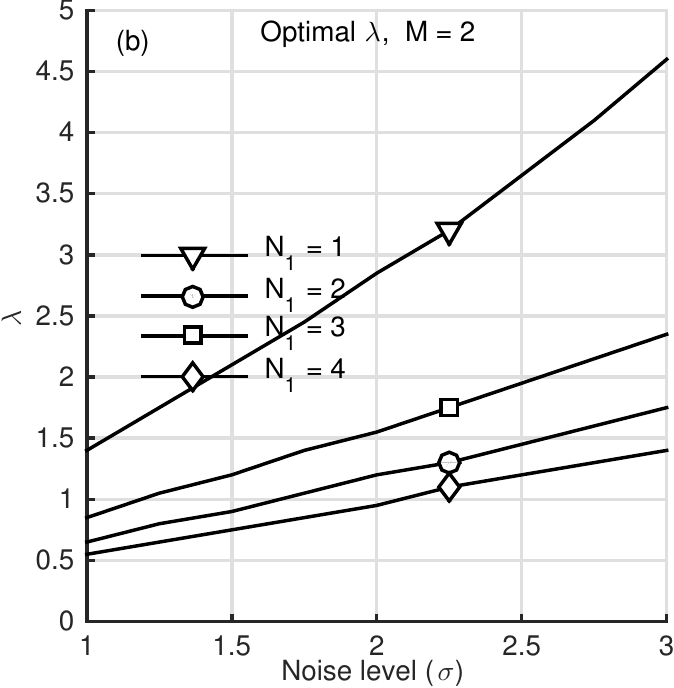}
    \\
   	\includegraphics[scale = \figurescale]{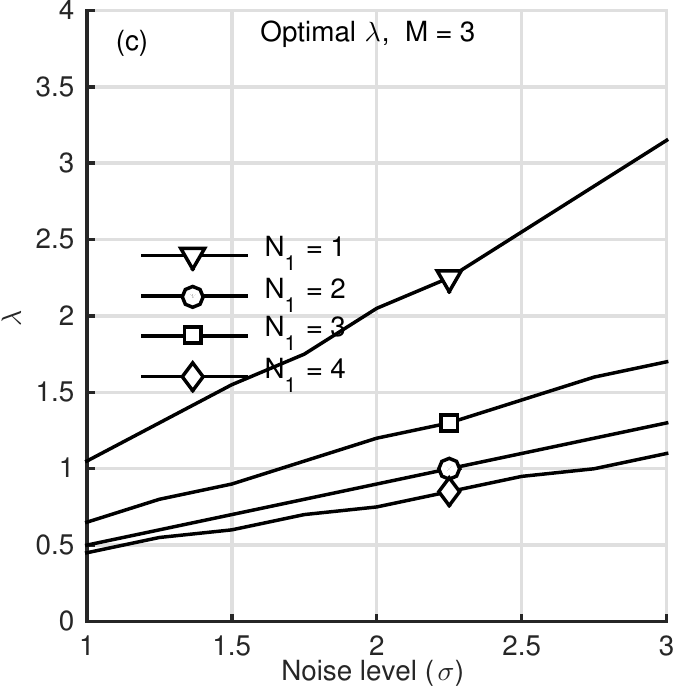}
	\quad
    \includegraphics[scale = \figurescale]{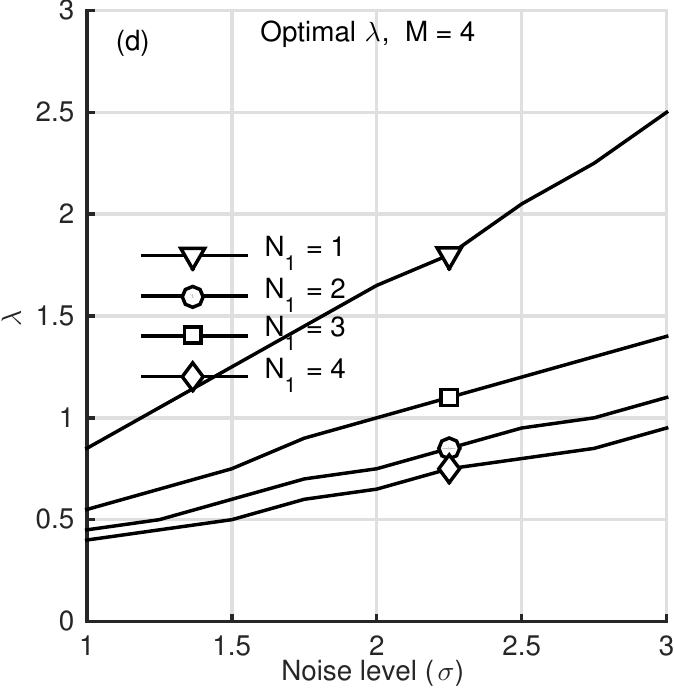}
	\caption{ Example 1: Optimal $\lambda$ at different noise level.
	(Note that the vertical axes of the sub-figures are different.)}
	\label{fig:example1_sigma}
\end{figure}


\textbf{Setting pattern $b$.}
In many cases of fault diagnosis,
it is feasible to estimate the period of the transients based by component geometry and rolling speed,
but it might be difficult to estimate the duration of the transient.
However, since the binary structure $b$ is acting as a sliding window capturing the global periodicity structure,
it is not necessary to match the length of a transient exactly.
As a consequence, we suggest to constrain the value of $N_1$ relatively small, e.g. $ 2 \le N_1 \le 4$.
However, if the sampling rate is quite high and a transient may contain more samples, then the value of $N_1$ can be specified greater than 4.
The value of $M$ determines the number of non-zero entries (the value of $K_1$) in $b$,
and $K_1$ effects the non-convexity of the regularizer in the proposed problem \eqref{eqn:fd_pogs}.
In our experiments, we keep use $M = 4$.

\medskip
\textbf{Setting regularization parameter $\lam$.}
In order to explore the correlation of the regularization parameter $\lam$ to the binary pattern $b$ \eqref{eqn:fd_preoid_b},
we show the optimal $\lambda$ as a function of the standard deviation $\sigma$ of noise using different binary patterns
in \figref{fig:example1_sigma}.
We define the optimal $\lam$ as the value minimizing the RMSE for each fixed binary pattern $b$.

In this test, using the data in \figref{fig:example1_data}(a),
we search for $\lam$ minimizing an average RMSE (20 trials) at each noise level generated by different random seeds.
We present the results of select $M$ from 1 to 4 in \figref{fig:example1_sigma}(a) to (d), respectively.
In each figure, we show the optimal $\lam$ as a function of $\sigma$ under different $N_1$.
Note that, in \figref{fig:example1_sigma}(a), when $M=1$ and $N_1 = 1$,
the proposed method adheres neither grouping nor periodic structure,
in which case, the problem \eqref{eqn:fd_pogs} coincides to the BPD problem with non-convex regularizer.

Note that since we simulate the test signal with a unit standard deviation,
the horizontal axis in \figref{fig:example1_sigma} is also the ratio of deviation from noise to data.
In other word, we can use \figref{fig:example1_sigma} as a look-up table,
when the input noise-to-signal ratio (SNR) is known.

Moreover, all the 16 curves in \figref{fig:example1_sigma}
show that the optimal $\lam$ varies approximately linearly with noise level.
In practice, we suggest to chose $\lam$ in \eqref{eqn:fd_pogs} proportional to the noise level as $\lam = r \sigma$.
Through further experiments, we provide Table~\ref{tab:fd_r} as a guide for choosing the multiplier $r$.

\begin{table}[!h]
	\caption{Selection of $r$ for setting $\lambda$.} 
	\label{tab:fd_r}
\begin{center}
    {
		\begin{tabular}{ l  c cccc }
		\toprule
		 \backslashbox{$M$}{$N_1$}	 &1   &2 &3 &4 \\[0.2em]
		\midrule
		~1 	 &3.700   &1.700 &1.150 &0.925	\\
		~2 	 &1.700   &0.850 &0.625 &0.475	\\
		~3 	 &1.150   &0.625 &0.450 &0.375	\\
		~4 	 &0.925   &0.475 &0.375 &0.325	\\[0.1em]
		\bottomrule
		\end{tabular}
	}
	\end{center}
\end{table}

\subsection{Receiver operating characteristic (ROC) evaluation}

\begin{figure}[!t]
	\centering
	\includegraphics[scale = \figurescale]{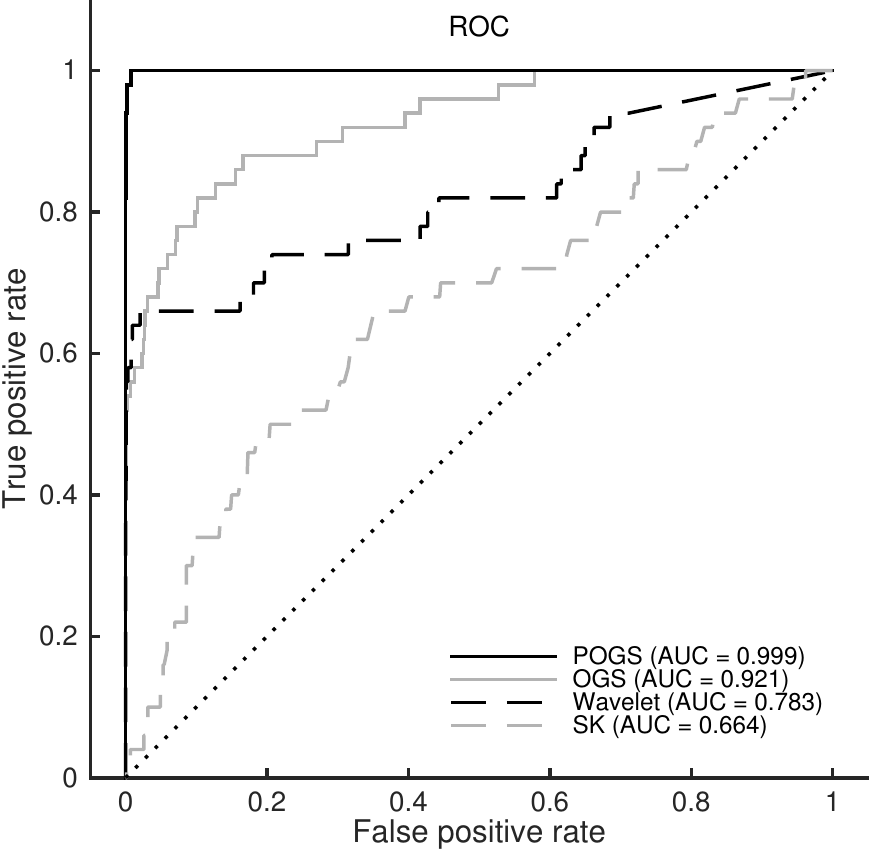}
	\caption{ Example 1: ROC curves of faults detection.}
	\label{fig:example_1_roc}
\end{figure}
\begin{figure}[!t]
	\centering
	\includegraphics[scale = 0.5]{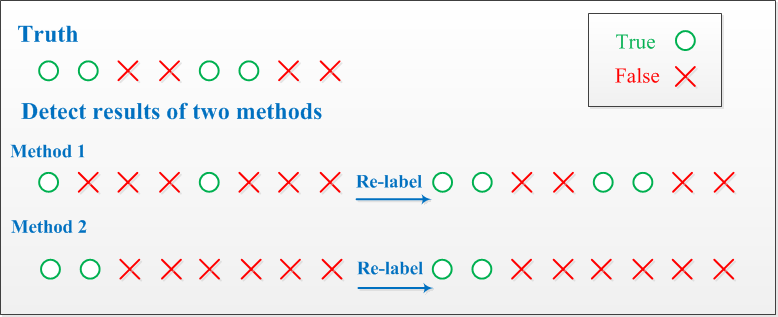}
	\caption{ Example of classification rule.}
	\label{fig:example_1_detect}
\end{figure}

The above comparative evaluation of OGS, POGS, wavelet, and fast spectral kurtosis uses RMSE as an indicator of denoising performance.
However, the RMSE by itself is not a sufficient indicator of fault detection accuracy.
In the following, we use a receiver operating characteristic (ROC) based approach to evaluate the relative accuracy of the methods.

The receiver operating characteristic (ROC) curve is a well known detection performance evaluation methodology \cite{lee1997higher}.
ROC curves are well suited to the problem of vibration-based diagnosis applications \cite{nichols2008using}.
An ROC curve is generated by plotting the probability of a false alarm against the probability of detection as the threshold level is varied.
Since the POGS approach is focused on machinery fault detection, the ROC curve is utilized to validate the superior detection performance of POGS compared to other methods.
We define the classification rule as:
if one sample in a transient [generated by \eqref{eqn:fd_transient}] is detected as positive,
then this entire transient a fault feature is detected and all the remaining samples are all assigned to be positive.

This rule is slightly different to the one used in \cite{nichols2008using},
since using the sample-wise decision rule in \cite{nichols2008using} may cause problem of overweighting sample recovery,
but neglecting detecting fault features as transients of a cluster of samples.
We use a simple example in \figref{fig:example_1_detect} to illustrate the problem and our classification rule.
Suppose that the true signal has 8 samples consists 2 periods (4 samples for each period),
and each period has 2 samples positive (labeled as circles).
If Method~1 detects 2 samples belonging to two different transients respectively,
and method 2 detects 2 samples belonging to only one transients,
then the rule of \cite{nichols2008using} will give that the two methods have identical accuracy,
because if merely counting the samples, they have a same number of samples detected.
This is undesirable, since Method~2 misses an entire transient.
To overcome this issue, we re-label the detect result,
where if any sample in a transient of each period is detected, we re-label all the samples in the entire transient to be detected.
In the example of \figref{fig:example_1_detect}, after the re-labeling, Method~1 has a better accuracy because it detects both of the transients.

\figref{fig:example_1_roc} shows the ROC curves from using OGS, POGS and wavelet-based method respectively.
POGS achieves an almost perfect detection result. Also, OGS is better than wavelet-based method.
In this example, because the results from POGS with different parameter settings
[in \figref{fig:example_1_results}(a) and \figref{fig:example_1_results}(b)] obtain almost identical ROC curves,
we chose to show the one in \figref{fig:example_1_results}(a) where $N_1=4$ and $M=4$.

We also use the extracted envelope from fast spectral kurtosis method to perform a similar ROC analysis.
Note that although we plot it together with other ROC curves in \figref{fig:example_1_roc},
since its envelope has a different length than the other results, the ROC curve is generated by a different number of samples.

\section{Experimental and engineering data validation}

\subsection{Example 2: Rolling bearing with defect on outer race}

\begin{table}[t]
  \centering
    \caption{The parameters of the tested bearing} \medskip
  \begin{tabular}{@{} ccccc @{}} 
    \toprule
    Inner Race & Outer Race & Roller & Number of roller & Contact angle \\
    \midrule
    160 mm & 290 mm & 34 mm & 17 & $0^{\degree}$ \\
    \bottomrule
  \end{tabular}
  \label{tab:example_3_parameters}
\end{table}

\begin{figure}[!t]
	\centering
    \includegraphics [scale = 0.8]{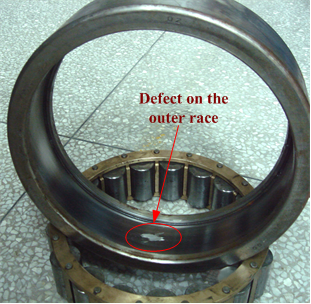}
	\caption{Example~2: Fault on the outer race of the testing bearing.}
	\label{fig:example_3_defect}
\end{figure}

\begin{figure}[!t]
	\centering
    \includegraphics [scale = \figurescale]{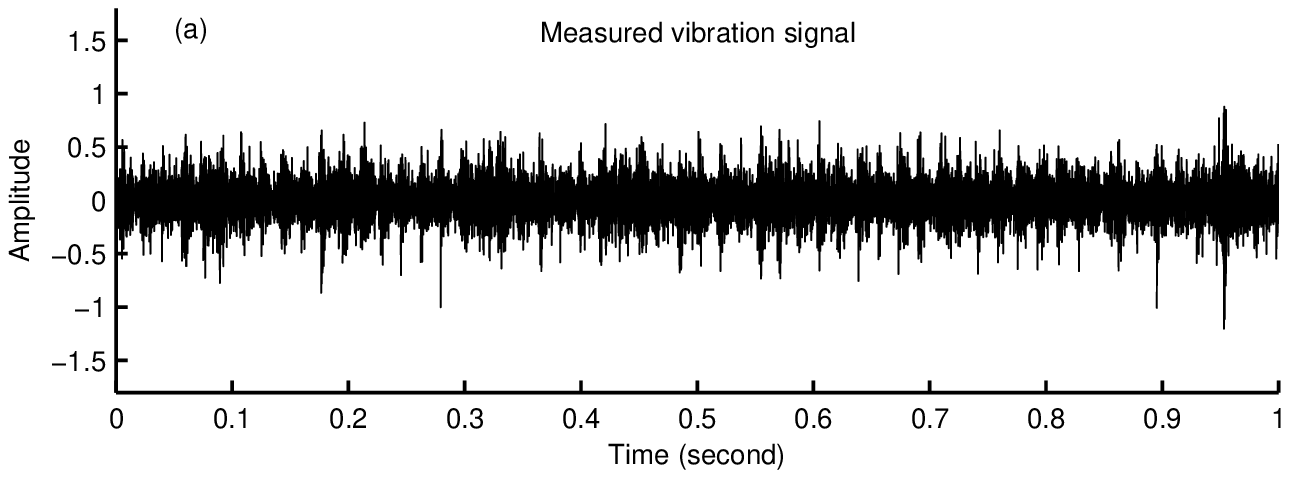}
    \\
    \includegraphics [scale = \figurescale]{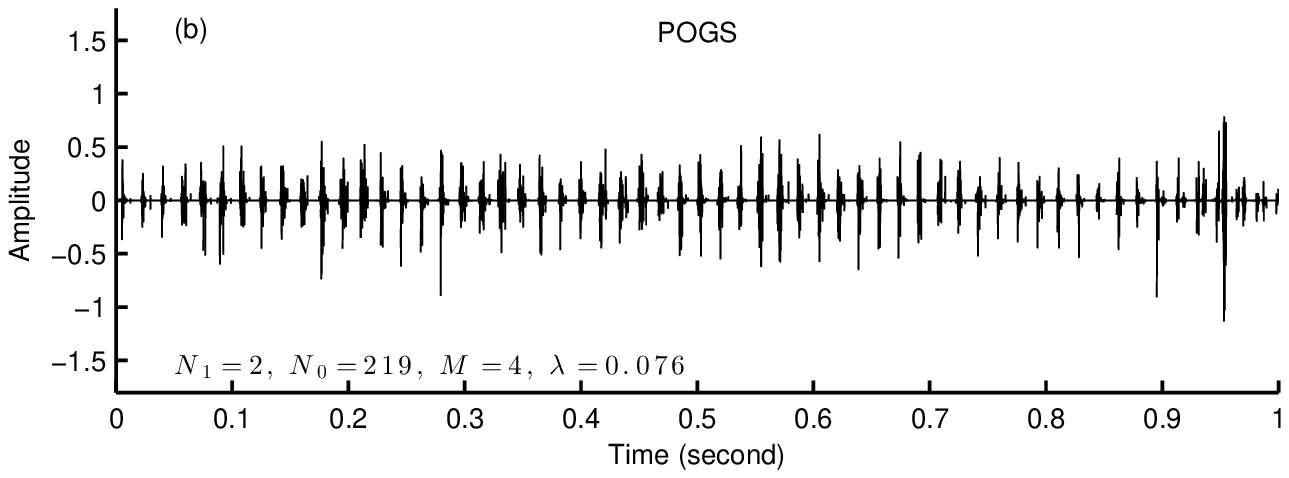}
    \\
    \includegraphics [scale = \figurescale]{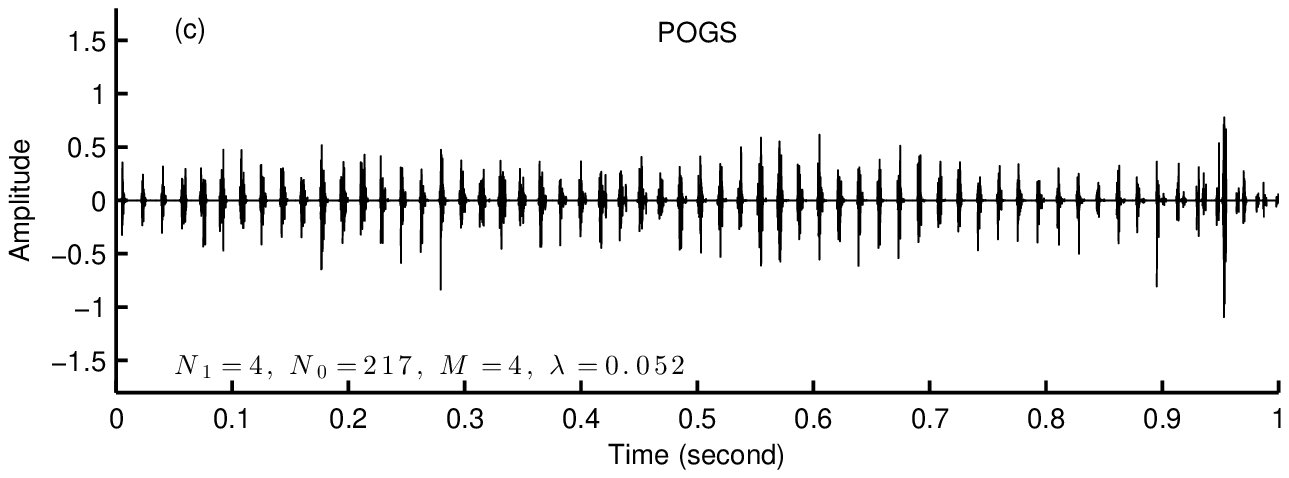}
    \\
    \includegraphics [scale = \figurescale]{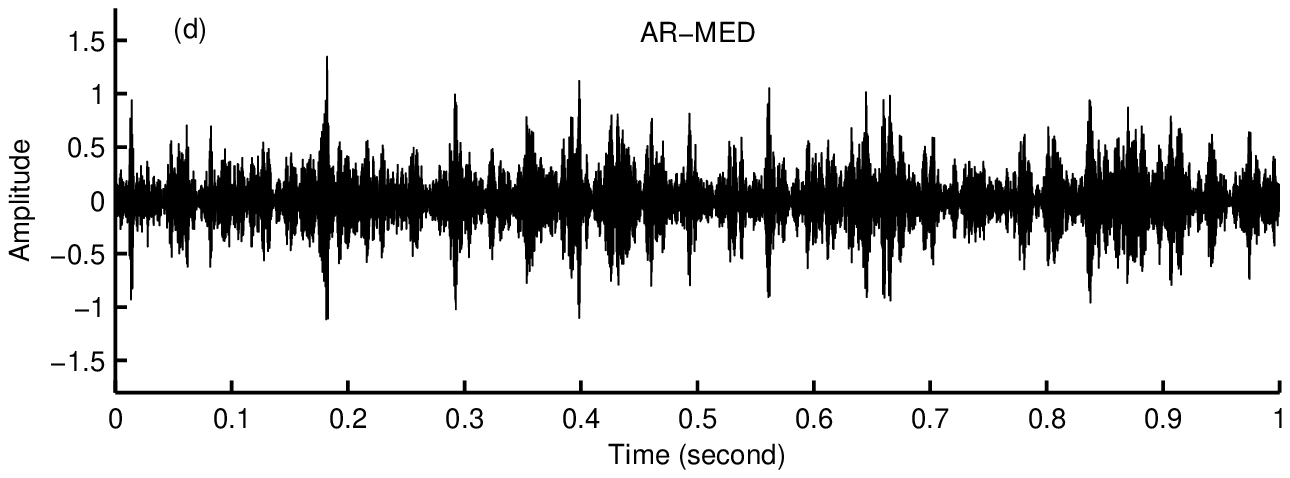}
	\caption{Example~2:  (a) Test data.
		(b) Result of proposed method with $N_1 = 2, M =4$.
		(c) Result of proposed method with $N_1 = 4, M =4$.
		(d) Result of AR-MED.}
	\label{fig:example_3_x}
\end{figure}

In this example, the proposed approach is applied to a vibration signal collected from a locomotive rolling element bearing with defect.
The testing locomotive rolling bearing with a slight scrape on its outer race is shown in \figref{fig:example_3_defect}.
The vibration signal is measured from acceleration sensors, using SONY Ex data acquisition system when the electric locomotive was running.
The bearing type is 552732QT and its specification is shown in \tabref{tab:example_3_parameters}.
The sampling rate is 12.8 kHz and the current rotating speed is approximately 481 rpm.
Thus, based on the geometric parameters and rotational speed, the characteristic frequency of the outer race defect is calculated to be $f_o = 57.8$ Hz.

To adopt the proposed POGS \eqref{eqn:fd_pogs}, 
if the duration of a transient is uncertain,
we can directly chose the pattern $b$ \eqref{eqn:fd_preoid_b} with $M =4$ and $N_1=2$,
then by \tabref{tab:fd_r} we can chose $\lambda$ with the respect to the noise level.
Note that since the `noise' in the vibration data for fault detection is the background out of the transient sequence,
in practice it can be easily estimated using healthy data.

In this example, we illustrate how to determine $\lam$ without healthy data.
We estimate the deviation of background by the formula
\begin{align}\label{eqn:mad}
	\hat \sigma = \mathsf{MAD}(y) / 0.6745
\end{align}
which is a conventional estimator of noise level used for wavelet-based denoising \cite{wav_Donoho_93_ideal},
where $ \mathsf{MAD}$ in \eqref{eqn:mad} stands for median absolute deviation, defined as
\begin{align}\label{eqn:mad_2}
	\mathsf{MAD}(y)  :=  \mathsf{median} ( \abs{  y_n - \mathsf{median}(y)  }).
\end{align}

In this example, the estimated deviation of background is $ \hat \sigma = 0.1606 $ obtained by \eqref{eqn:mad}.
The parameter $\lam$ can be obtained using $\lam = r \hat \sigma$ where $r = 0.475$ (from Table~\ref{tab:fd_r}).
The result is shown in \figref{fig:example_3_x}(b), where a periodical phenomenon can be easily identified.
Moreover, we also test POGS with $M = 4$ and $N_1 = 4$, using the same method to set $\lam$.
The result is shown in \figref{fig:example_3_x}(c).

As a comparison, we also test the data with autoregression model assisted minimum entropy deconvolution (AR-MED) \cite{fd_Endo_2007}\footnote{
Implementation available online \url{http://www.mathworks.com/matlabcentral/fileexchange/41614-ar-filter-+-minimum-entropy-deconvolution-for-bearing-fault-diagnosis}}.
The order of AR filter is chosen by maximizing Kurtosis value as the implementation suggested,
and the estimated filter in the MED step has a length of 50.
The result is shown in \figref{fig:example_3_x}(d), where most of the impulses are promoted after deconvolution.
However, since the baseline is still relatively noisy, the regularity of periodicity is not very clear.

\subsection{Example 3: Motor bearing with multiple faults}

\begin{figure}[!t]
	\centering
    \includegraphics [scale = 0.4]{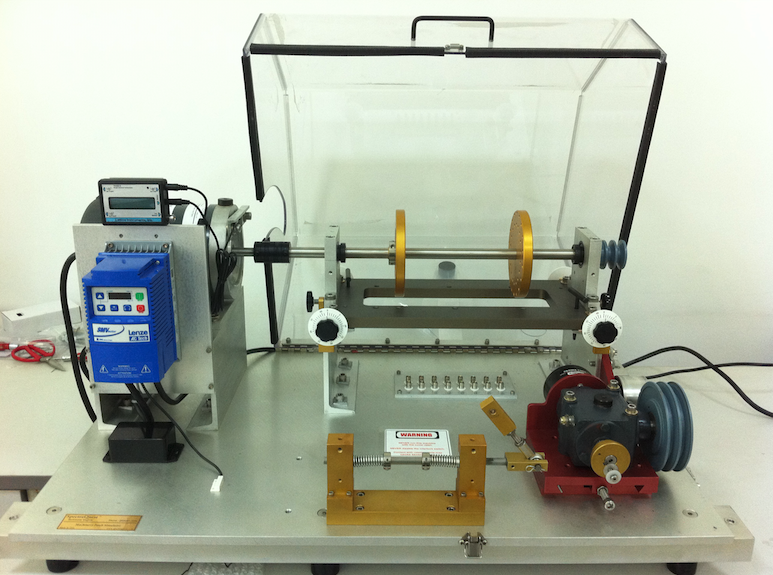}
	\caption{Example~3: Spectra Quest's machinery fault simulator.}
	\label{fig:Experimental_setup}
\end{figure}
\begin{figure}[!t]
	\centering
    \includegraphics [scale = \figurescale]{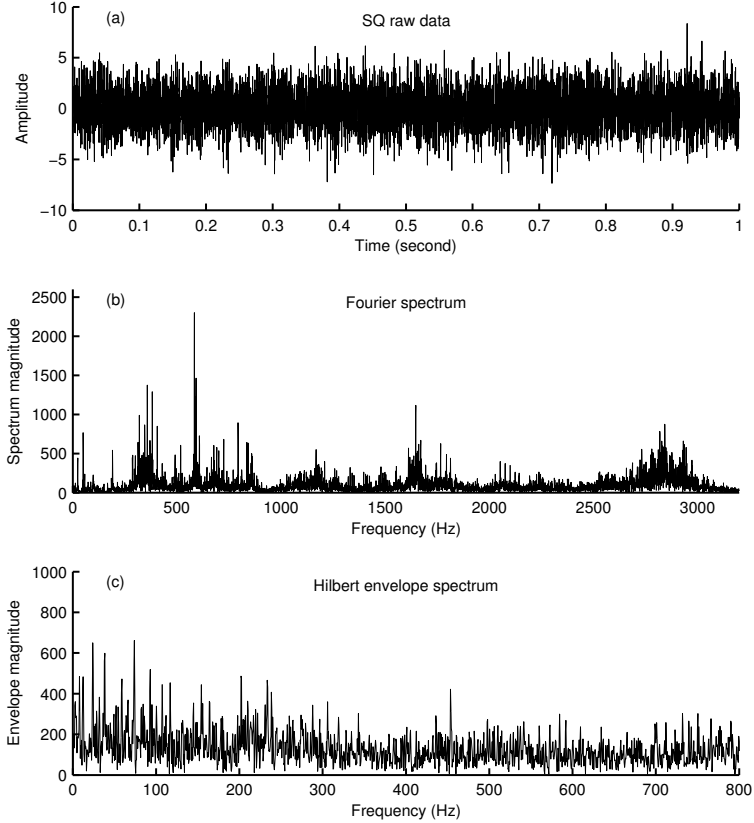}
	\caption{Example 3: Test data.
		(a) The measured vibration signal of the motor housing.
		(b) Fourier spectrum.
		(c) Hilbert envelope spectrum.
	}
	\label{fig:example_2_y}
\end{figure}

\begin{figure}[!t]
	\centering
    \includegraphics [scale = \figurescale]{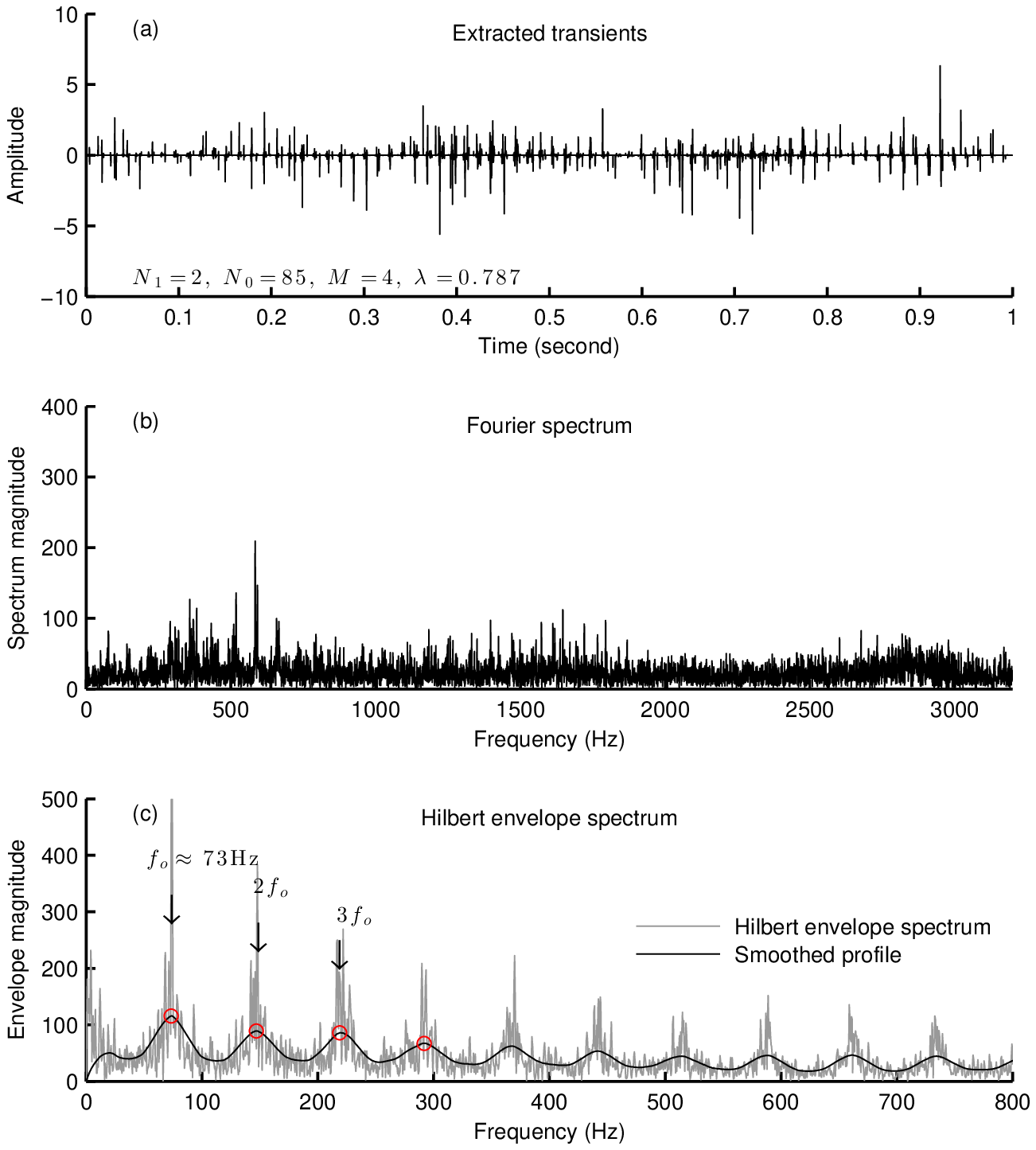}
	\caption{Example 3: Outer race faults detection ($N_1 = 2, M = 4 $): (a) Extracted transients, (b) Fourier spectrum and (c) Hilbert envelope spectrum.}
	\label{fig:example_2_x1}
\end{figure}

\begin{figure}[!t]
	\centering
    \includegraphics [scale = \figurescale]{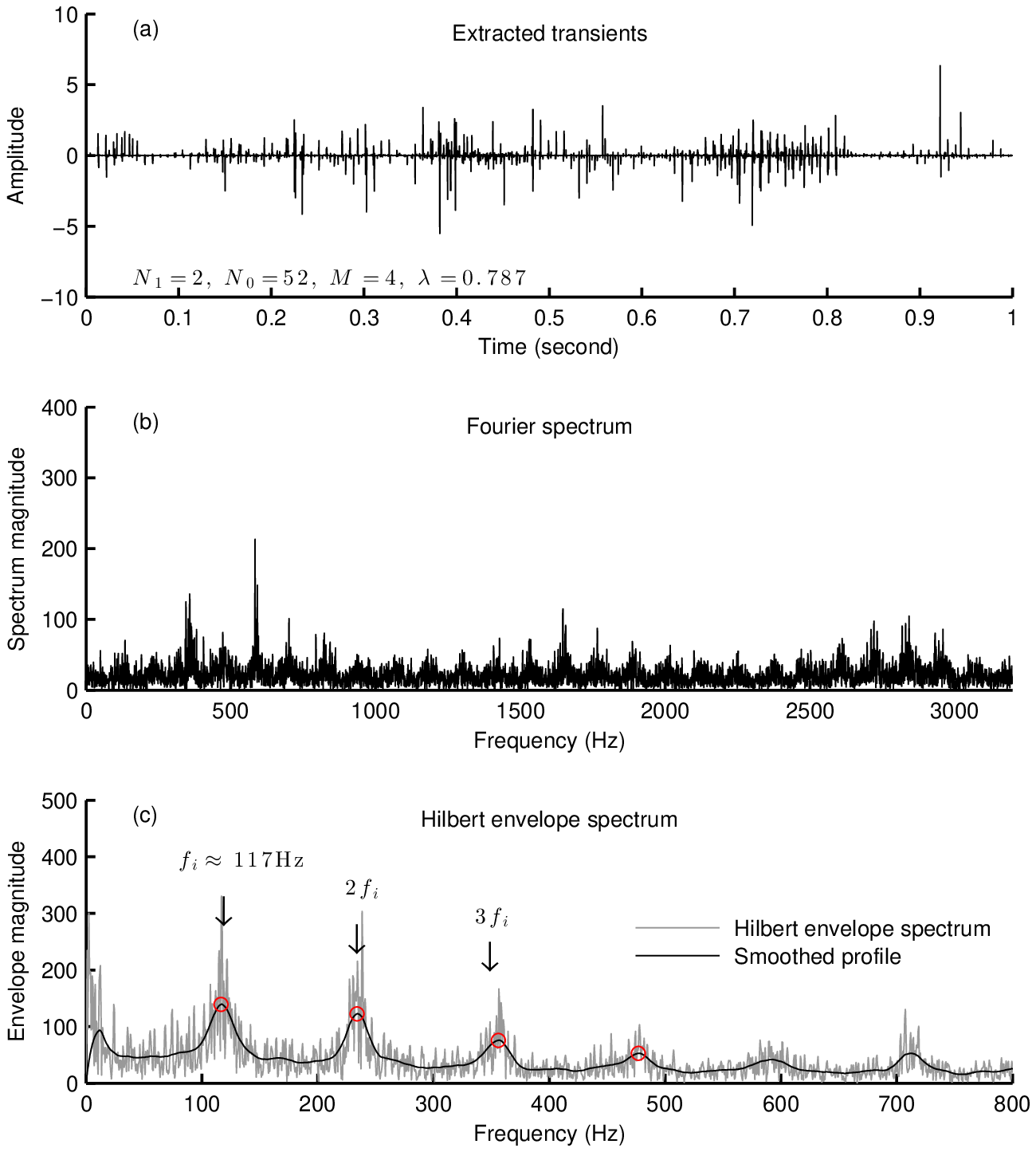}
	\caption{Example 3: Inner race faults detection ($N_1 = 2, M = 4 $): (a) Extracted transients, (b) Fourier spectrum and (c) Hilbert envelope spectrum.}
	\label{fig:example_2_x2}
\end{figure}

\begin{figure}[!t]
	\centering
	\quad
    \includegraphics[scale = \figurescale]{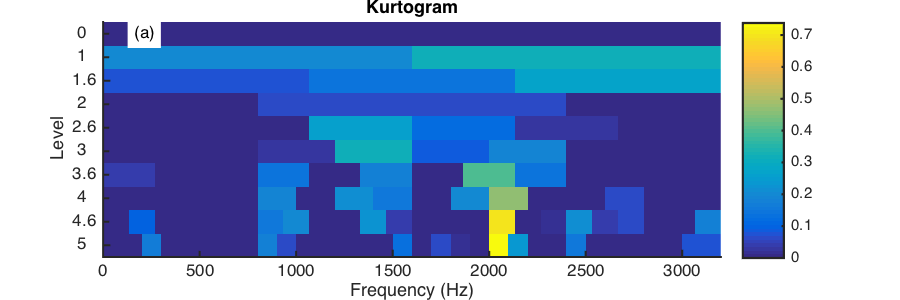}
	\\
    \includegraphics[scale = \figurescale]{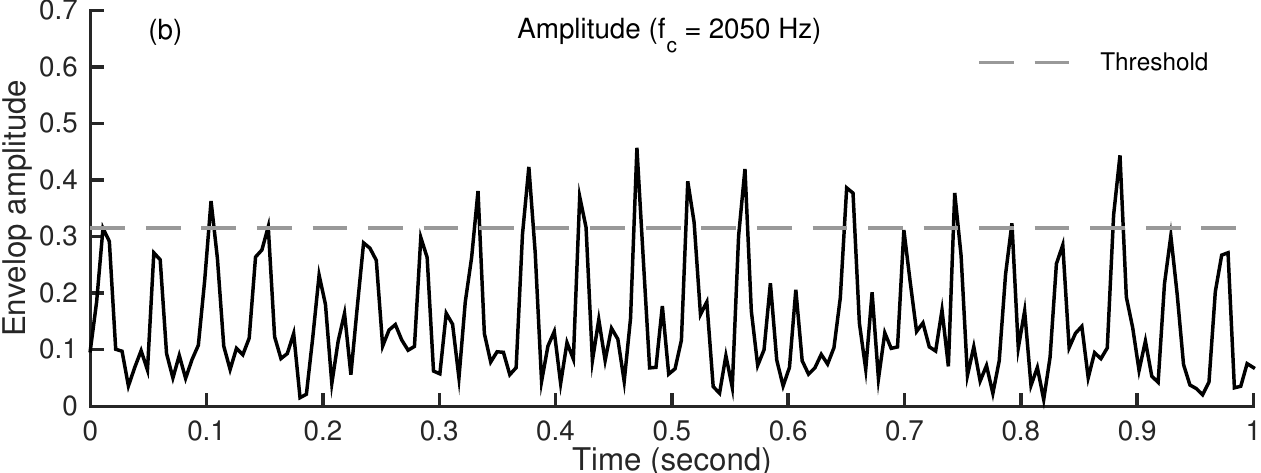}
	\caption{ Example 3: Results of fast spectral kurtosis.
	(a) Kurtogram. (b) Envelope of extracted transients.}
	\label{fig:example_2_kurt}
\end{figure}

\begin{table}[htbp]
  \centering
    \caption{Fault frequencies for MFS motor bearing } \medskip
  \begin{tabular}{@{} lcccr @{}} 
    \toprule
    Component & FTF & BPFO & BPFI & BSF \\
    \midrule
    Motor Bearings & 0.384 & 3.066 & 4.932 & 2.03 \\
    \bottomrule
  \end{tabular}
  \label{tab:fault_frequency}
\end{table}

To further demonstrate the effectiveness of the proposed method for machinery fault detection,
applications of motor bearing fault diagnosis are studied in this section.
The experiment is performed on a Spectra Quest's machinery fault simulator (MFS) illustrated in \figref{fig:Experimental_setup}.
The user does not need to make any modifications to the motor provided.
The simulation setup consists of a motor with intentionally faulted bearings:
one bearing with an inner race fault, and one bearing with an outer race fault.
Therefore, the fault diagnosis of the motor bearings is equivalent to a compound faults detection case.
The motor bearing fault frequencies for MFS components are given in Table~\ref{tab:fault_frequency}.
Accelerometers were mounted on the motor housing. The vibration signals were measured at a sampling rate $ f_s =  6400 $ Hz.
The rotating speed of the motor is 1433 rpm (23.89 Hz).
Hence, the characteristic frequencies of the inner race and outer race of the motor bearings are calculated to be $f_i \approx 117.8$ Hz and $f_o \approx 73.2$ Hz, respectively.

An observed vibration signal with a duration of 1 second is illustrated in \figref{fig:example_2_y}(a). However, the useful periodic pulses are buried in strong background noise and irrelevant interference.
The frequency spectrum and the Hilbert envelope spectrum of the signal are shown in \figref{fig:example_2_y}(b) and (c), respectively.
It can be observed from \figref{fig:example_2_y}(b) that the energy of the spectrum is distributed along the whole frequency range. Peaks at 24 Hz and its harmonics can be observed in the low frequency band, which correspond to the rotating frequency and its harmonics. The useful characteristic frequencies used to monitor the health status of motor bearings can not be observed in \figref{fig:example_2_y}(c).

The proposed POGS approach is utilized to process the vibration signal.
The results are shown in \figref{fig:example_2_x1} and \figref{fig:example_2_x2}.
We run POGS twice with two group structures,
determined by the inner race period $T_{i}$ and the outer race period $T_{o}$,
for the purpose of separating the useful impulsive fault features.
In particular, keeping $ N_1 = 2 $ and $ M = 4 $,
we define two different group structures using \eqref{eqn:fd_preoid_K} and  \eqref{eqn:fd_preoid_b}
based on the inner race period $T_{i} = 1/f_i$ and the outer race period $T_{o} = 1/f_o$, respectively.
We use the atan penalty function with $ a = 1/ (K_1\lam)$ to ensure convexity of the objective function.
The value of $\lam$ obtained using healthy data, is set so as to diminish the healthy data to almost all zeros.

The two periodic-related values correspond to the outer race defect and the inner race defect frequencies respectively.
Strong periodic impulses with intervals of approximately 0.0133 second (75 Hz) are clearly revealed \figref{fig:example_2_x1}(a),
which is exactly in accordance with the outer race characteristic frequency of 73.2 Hz.
Similarly, periodic transient features with the period 0.0085 second (118 Hz) can be observed in \figref{fig:example_2_x2}(a),
which is approximately the inner race characteristic frequency of 117.8 Hz.
To further reveal the characteristic frequencies, the frequency spectrum and the Hilbert envelope spectrum of the processed signals are shown in \figref{fig:example_2_x1}(b) and (c) and \figref{fig:example_2_x2}(b) and (c).
We also present the smoothed profiles of the Hilbert envelop spectrum to indicate the characteristic frequencies more clearly.
The Hilbert spectrum of the processed result illustrated in \figref{fig:example_2_x1}(c) is obtained using the proposed approach with prior knowledge of the outer race characteristic frequency 73.2 Hz.
Apparently, the characteristic frequencies of outer race 73 Hz and its harmonic components are clearly revealed, as shown in \figref{fig:example_2_x1}(c). 
Similarly, \figref{fig:example_2_x2}(c) is obtained utilizing the proposed approach with prior knowledge of the outer race characteristic frequency 117.8 Hz. 
The characteristic frequencies of inner race 117 Hz and its harmonic components can be clearly observed in \figref{fig:example_2_x2}(c). Thus, the proposed periodic non-convex regularized OGS approach successfully detects the compound faults of the motor bearings. 
More specifically, the fault features of outer race defect and inner race defect are clearly separated utilizing the proposed approach.

To further demonstrate the effectiveness of the proposed approach,
we also processed the vibration signal using spectral kurtosis and the results are presented in \figref{fig:example_2_kurt}.
\figref{fig:example_2_kurt}(a) is the fast kurtogram, where the optimal carrier is detected at $2050$ Hz.
Under this frequency,
periodic transients can be observed in the envelope of the filtered signal, in accordance with the rotating speed of motor (23.89 Hz), as shown in Fig.\ref{fig:example_2_kurt}(b).
No further fault-related information can be observed in \figref{fig:example_2_kurt}(b).

\section{Conclusion}
This paper proposes a periodic group sparsity approach for the purpose of detecting faults in rotating machinery.
The approach uses non-convex penalty functions to promote periodic group sparsity. We show how to constrain the non-convex penalty functions to ensure that the objective function is convex.
The OGS method was introduced in \cite{Chen_Selesnick_2014_OGS} and extended to non-convex regularized OGS in \cite{Chen_Selesnick_2014_GSSD}.
A novelty of the proposed approach is that the proposed penalty function models the periodicity of the sparse groups, making it suitable specifically for feature extraction in machinery fault diagnosis.
The period of the sparse pulses is chosen based on prior knowledge of machine geometry under inspect.
Moreover, the proposed approach is able to separate compound fault features by utilizing different periods of the periodic pulses corresponding to different fault frequencies (e.g., outer race and inner race characteristic frequencies of rolling element bearings) as demonstrated in Section~5.
The effectiveness of the proposed approach is verified by simulation and experimental data.
The processed results demonstrate that the proposed approach outperforms other methods.
%

%

\section*{Appendices}
\appendix
\renewcommand*{\thesection}{\Alph{section}}

\section{Proof of Proposition~\ref{prop:prop_1}}\label{app:A}
\noindent
\textit{Proof.~}
This proposition can be proven by taking the second-order derivative, where $\phi(v;a)$ is twice differentiable when $ v \neq 0$.
The second-order derivative of \eqref{eqn:fd_pv} is
\begin{align}
	p''(v) = \gamma + \lam \phi''(v;a), \quad v\neq 0.
\end{align}
Therefore, when $ v \neq 0 $, it is sufficient that, if $ \gamma + \lam \phi''(v;a) > 0 $, then $ p''(v) > 0$.
When $v=0$, Lemma~A in \cite{Chen_Selesnick_2014_GSSD} can be utilized directly, to show that,
since $ p'( 0^{-}) <  p'( 0^{+}) $, under the condition \eqref{eqn:fd_phi_condition}, function $p$ in \eqref{eqn:fd_pv} is strictly convex on $\real$.

\section{Proof of Proposition~\ref{prop:prop_2}}\label{app:B}
\noindent
\textit{Proof.~}
We rewrite $\sum_{ k} b_k u_{k}^2 $ with the respect to \eqref{eqn:fd_Ksets} as
\begin{align}
	\sum_{ k \in \Kset} b_k u_{k}^2 = \sum_{ k \in \Kset_1 } b_k u_k^2  + \sum_{ k \in \Kset_0 } b_k  u_k^2.
\end{align}
Then \eqref{eqn:fd_ogs_Pu} is given by
\begin{subequations}
\begin{align}
	P(u) 	&  = \frac{1}{2 K_1} \Big[ \sum_{ k \in \Kset_1 } b_k u_k^2  + \underbrace{ \sum_{ k \in \Kset_0 } b_k  u_k^2 }_{ 0 } \Big]
				+ \lam \phi \bigg( \Big[ \sum_{ k \in \Kset_1 } b_k  u_k^2  + \underbrace{ \sum_{ k \in \Kset_0 } b_k  u_k^2  }_{ 0 } \Big] ^{1/2} ;a \bigg) \\
			& = \frac{1}{2 K_1} \sum_{ k \in \Kset_1 } u_k^2 + \lam \phi\bigg( \Big[\sum_{ k \in \Kset_1 }u_{k}^2 \Big]^{1/2} ;a \bigg).
\end{align}
\end{subequations}
%
Using Proposition~\ref{prop:prop_1} with $v= [ \sum_{ k \in \Kset_1} u_{k}^2 ]^{1/2}$ and $\gamma = 1/K_1$,
it follows that, if \eqref{eqn:fd_ogs_prop_2} is satisfied, then $P$ is strictly convex.


\end{document}